\documentclass[a4paper,10pt,twocolumn]{article}

\usepackage[dvipsnames]{xcolor}
\usepackage{times}
\usepackage{amsmath}
\usepackage{graphicx}
\usepackage{subcaption}
\usepackage{xspace}
\usepackage[utf8]{inputenc}
\usepackage{ifthen}
\usepackage{balance}

\newboolean{extendedtr}
\setboolean{extendedtr}{true}

\newcommand{\extendedtr}[2]{%
\ifextendedtr
#1
\else
#2
\fi
}

\ifextendedtr
\usepackage[total={7in,9.5in}]{geometry}
\else
\fi

\usepackage{hyperref}
\hypersetup{
    colorlinks,
    linkcolor={red!50!black},
    citecolor={blue!50!black},
    urlcolor={blue!80!black}
}

\usepackage{cleveref}
\crefname{section}{§}{§§}
\Crefname{section}{§}{§§}

\graphicspath{ {./figures/} }

\usepackage{amsfonts}
\usepackage{amsmath}
\usepackage{amssymb}
\usepackage{pifont}

\usepackage[inline]{enumitem}
\usepackage[labelfont=bf]{caption}

\ifextendedtr
\newtheorem{theorem}{Theorem}
\newtheorem{proof}{Proof}
\else
\fi

\usepackage[utf8]{inputenc}
\usepackage[utf8]{luainputenc}

\makeatletter
\def\blx@maxline{77}
\makeatother

\renewcommand\paragraph[1]{\vspace{.6\baselineskip}{\bf #1}.}
\newcommand\smallparagraph[1]{\vspace{.3\baselineskip}{\bf #1}.}

\newcommand{\eg}{{e.g.},\xspace}
\newcommand{\ie}{{i.e.},\xspace}

\newcommand\code[1]{\texttt{#1}}

\newcommand\theTR{extended technical report~\cite{kpg-tr}\xspace}

\newcommand{\sys}{DD\xspace}

\newcommand{\rom}[1]{(\textit{\lowercase\expandafter{\romannumeral #1\relax})}}

\newcommand{\arrange}{\texttt{arrange} }

\definecolor{ao}{rgb}{0.0, 0.75, 0.0}
\definecolor{orange}{rgb}{0.88, 0.46, 0.25}

\ifextendedtr
\else
\vldbTitle{Shared Arrangements: practical inter-query sharing for streaming dataflows}
\vldbAuthors{Frank McSherry, Andrea Lattuada, Malte Schwarzkopf, Timothy Roscoe}
\vldbDOI{https://doi.org/10.14778/3401960.3401974}
\vldbVolume{13}
\vldbNumber{10}
\vldbYear{2020}
\fi

\newcommand{\seclabel}[1]{\label{sec:#1}}
\newcommand{\secref}[1]{\S\ref{sec:#1}}
\newcommand{\subseclabel}[2]{\label{sec:#1:#2}}
\newcommand{\subsecref}[2]{\S\ref{sec:#1:#2}}

\newcommand{\tableref}[1]{Figure~\ref{#1}\xspace}

\ifextendedtr
\else
\fi

\title{Shared Arrangements: practical inter-query sharing\\for streaming dataflows}

\begin{document}

\ifextendedtr
\date{}
        \author{
        Frank McSherry\textsuperscript{$\ast$}\qquad Andrea Lattuada\qquad Malte Schwarzkopf\textsuperscript{\ddag}\qquad Timothy Roscoe\\
        \textsuperscript{$\ast$}Materialize, Inc.\qquad Dept.\ of Computer Science, ETH Zürich\qquad \textsuperscript{\ddag}Brown University
        }
\else
\numberofauthors{1}
        \author{
        \alignauthor
        Frank McSherry\textsuperscript{$\ast$}\qquad Andrea Lattuada\qquad Malte Schwarzkopf\textsuperscript{\ddag}\qquad Timothy Roscoe\\
               \affaddr{\textsuperscript{$\ast$}Materialize, Inc.\qquad Dept.\ of Computer Science, ETH Zürich\qquad \textsuperscript{\ddag}Brown University}\\
               \email{mcsherry@materialize.io, \{andreal,troscoe\}@inf.ethz.ch, malte@cs.brown.edu}
        }
\fi

\maketitle

\section*{Abstract}

Current systems for data-parallel, incremental processing and
view maintenance over high-rate streams isolate the execution
of independent queries.
This creates unwanted redundancy and overhead in the presence of
concurrent incrementally maintained queries:
each query must independently maintain the same
indexed state over the same input streams, and new queries must build
this state from scratch before they can begin to emit their first
results.

This paper introduces \textit{shared arrangements}: indexed views of
maintained state that allow concurrent queries to reuse the same
in-memory state without compromising data-parallel performance and
scaling.
We implement shared arrangements in a modern stream processor and show
order-of-magnitude improvements in query response time and resource
consumption for interactive queries against high-throughput streams, while
also significantly improving performance in other domains including
business analytics, graph processing, and program analysis.


\section{Introduction}
\seclabel{intro}

In this paper, we present \emph{shared arrangements}, a new technique
for efficiently sharing indexed, consistent state and computation
between the operators of multiple concurrent, data-parallel streaming
dataflows.
We have implemented shared arrangements in \sys, the current
implementation of Differential Dataflow~\cite{naiad,
differential-dataflow, differential}, but
they are broadly applicable to other streaming systems.

Shared arrangements are particularly effective in interactive data
analytics against continually-updating data.
Consider a setting in which multiple analysts, as well as software like
business intelligence dashboards and monitoring systems, interactively
submit standing queries to a stream processing system.
The queries remain active until they are removed.
Ideally, queries would install quickly, provide initial results
promptly, and continue to deliver updates with low latency as the
underlying data change.

Data-parallel stream processors like Flink~\cite{flink}, Spark Streaming~\cite{spark-streaming}, and Naiad~\cite{naiad} excel at incrementally maintaining the results of such queries, but each maintain queries in independent dataflows with independent computation and operator state.
Although these systems support the sharing of common sub-queries, as streams of data, none share the \emph{indexed} representations of relations among unrelated subqueries.
However, there are tremendous opportunties for sharing of state, even when the dataflow operators are not the same.
%
%
For example, we might expect joins of a relation $R$ to use its primary key; even if several distict queries join $R$ against as many other distinct relations, a shared index on $R$ would benefit each query.
Existing systems create independent dataflows for distinct queries, or are restricted to redundant, per-query indexed representations of $R$, wasting memory and computation.


\begin{figure*}[t]
\begin{subfigure}[t]{0.329\textwidth}
\includegraphics[trim={0 0 0 0}, clip, width=\textwidth]{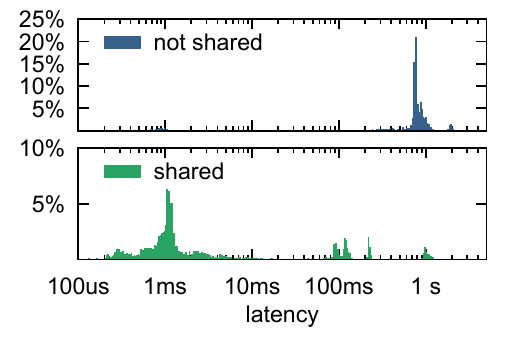}
\caption{Query installation latency.}
\label{fig:eval:tpch-hist-install}
\end{subfigure}
~
\begin{subfigure}[t]{0.329\textwidth}
\includegraphics[trim={0 0 0 0}, clip, width=\textwidth]{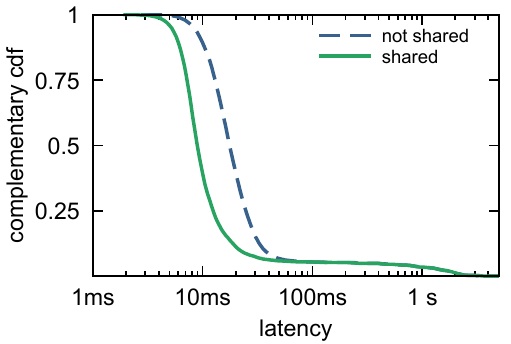}
\caption{Update processing latency.}
\label{fig:eval:tpch-hist-work}
\end{subfigure}
~
\begin{subfigure}[t]{0.329\textwidth}
\includegraphics[trim={0 0 0 0}, clip, width=\textwidth]{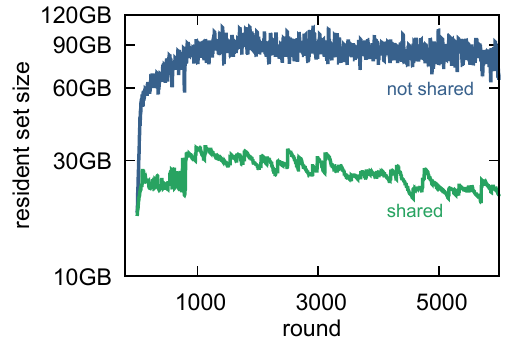}
\caption{Memory footprint (RSS).}
\label{fig:eval:tpch-rss}
\end{subfigure}
\caption{Shared arrangements reduce (\subref{fig:eval:tpch-hist-install}) query installation latency distribution, (\subref{fig:eval:tpch-hist-work}) update processing latency distribution, and (\subref{fig:eval:tpch-rss}) the memory footprint of concurrent TPC-H queries that randomly arrive and retire. The setup uses 32 workers, runs at TPC-H scale factor 10, and loads rows from relations round-robin. Note the $\log_{10}$-scale $x$-axes in (\subref{fig:eval:tpch-hist-install}) and (\subref{fig:eval:tpch-hist-work}), and the $\log_{10}$-scale $y$-axis in (\subref{fig:eval:tpch-rss}).}
\label{fig:eval:tpch_maintain}
\end{figure*}

By contrast, classic relational databases have long shared indexes over their tables across unrelated queries.
The use of shared indexes reduces query times tremendously, especially for point look-ups, and generally improves the efficiency of queries that access relations by the index keys.
While they have many capabilities, relational databases lack streaming dataflow system's support for low-latency, high-throuhput incremental maintenance of materialized query results~\cite{noria, dbtoaster}.
Existing shared index implementations share all reads and writes among multiple workers, and are not immediately appropriate for dataflow workloads where the operator state is sharded across independent workers.
In this work, we seek to transport the shared index idiom from relational databases to streaming dataflows, applying it across changing maintained queries.

Our main observations are that \rom{1} many dataflow operators write the same internal state, representing the accumulated changes of each of their input streams, \rom{2} these dataflow operators often access this state with independent and fundamentally different patterns, and \rom{3} this state can be efficiently shared with single-writer, multiple-reader data structure. Shared arrangements are our design for single-writer, multiple-reader, shared state in dataflow systems.

To illustrate a natural setting for shared arrangements, we run a mix of interactively issued and incrementally maintained TPC-H~\cite{tpch} queries executed as dataflows against a stream of order fulfillment events (\ie changes to the \code{lineitem} relation).
This is similar to a modern business analytics setting with advertisers, impressions, and advertising channels, and our dynamic query setup mimicks the behavior of human analysts and business analytics dashboards.\footnote{TPC-H is originally a static ``data-warehousing'' benchmark;
our streaming setup follows that used by Nikolic et al.~\cite{hotdog}.}
We measure the query installation latency---\ie the time until a new query returns results---as well as update processing latency and standing memory footprint.
Figure~\ref{fig:eval:tpch_maintain} reports the performance of \sys with shared arrangements (``shared'') and without (``not shared''; representative of other data-parallel stream processors).
The measurements show orders of magnitude improvements in query installation latency (a weakness of existing dataflow systems), and improved update processing latency and memory use.
Shared arrangements achieve these improvements because they remove the need to maintain dataflow-local indexes for each query.
As a concrete example throughout this paper, we consider TPC-H queries 3 and 5.
%
%
Both queries join \code{lineitem} with the \code{order} and \code{customer} relations by their primary keys. 
While the queries lack overlapping subqueries that classic multi-query optimization (MQO) would detect, they both perform lookups into \code{order} and \code{customer} by their respective primary keys when processing an updated \code{lineitem} record.
%
%
Existing stream processors will create and maintain a per-query index for each relation, as these systems are designed to decouple the execution of dataflow operators.
%
%
Shared arrangements, by contrast, allow Q3 and Q5 to share indexes for these two relations.
This can dramatically reduce the time to install the second query and provide initial results, and also increases overall system capacity, as multiple queries share in-memory indexes over the same relations.
Finally, these benefits come without restricting update throughput or latency, as they do not change the data-parallel execution model of the stream processor.
%




\medskip

%

The key challenge for shared arrangements is to balance the opportunities of sharing against the need for coordination in the execution of the dataflow.
In the scenarios we target, logical operator state is sharded across multiple physical operators; sharing this state between the operators of multiple queries could require global synchronization.
Arrangements solve this challenge by carefully structuring how they share data: they \rom{1} hard-partition shared state between worker threads and move computation (operators) to it, and \rom{2} multiversion shared state within workers to allow operators to interact with it at different times and rates.

Our full results in~\secref{eval} confirm that shared arrangements translate into two benefits: \rom{1} queries deploy and produce correct results immediately without rescanning historical data, and \rom{2} the same capacity (stream volume and concurrent queries) can be achieved with fewer cores and less RAM.
For a streaming variant of TPC-H and a changing graph, shared arrangements also reduce update latency by {1.3--3$\times$} and reduce the memory footprint of the computation by {2--4$\times$}, compared to systems that do not share indexed state.
These benefits hold without degrading performance on other tasks---batch and interactive graph processing, and Datalog-based program analysis---on which \sys outperforms other systems.

Shared arrangements can be applied to many modern stream processors,
but we implemented them as part of \sys.
%
\sys has been the publicly available reference implementation of Differential Dataflow for several years~\cite{differential}, and is deployed in variety of industrial settings.
For example, VMware Research uses \sys to back their reactive DDlog Datalog engine~\cite{ddlog}, applied to problems in network reconfiguration and program analysis.
Shared arrangements have proved key to the system's success.
%




%
Some benefits of shared arrangements are attainable in purely windowed streaming settings, which ensure that only bounded historical state must be reviewed for new queries.
However, shared arrangements provide similar benefits without these restrictions, and support windowing of data as one of several join idioms.
The main limitation of shared arrangements is that their benefits apply only in the cases where actual sharing occurs; while sharing appears common in settings with relational data and queries, bespoke stream processing computations (\eg with complex and disjoint windowing on relations) may benefit to varying and lesser degrees.

In many ways, shared arrangements are the natural interpretation of an RDBMS index for data-parallel dataflow, and bring its benefits to a domain that has until now lacked them.




\section{Background and Related Work}
\seclabel{background}



\begin{figure*}
  \centering
  \begin{tabular}{r|r|r|r|r}
    \bf System class           & \bf Example      & \bf Sharing               & \bf Updates            & \bf Coordination  \\ \hline
    RDBMS                      & Postgres         & \textbf{Indexed state}    & \textbf{Record-level}  & Fine-grained              \\ \hline
    Batch processor            & Spark            & Non-indexed collections   & Whole collection       & \textbf{Coarse-grained} \\ \hline
    Stream processor           & Flink            & None                      & \textbf{Record-level}  & \textbf{Coarse-grained} \\ \hline
    Shared arrangements        & \sys             & \textbf{Indexed state}    & \textbf{Record-level}  & \textbf{Coarse-grained} \\
  \end{tabular}
  \caption{Sharing of indexed in-memory state, record-level update granularity,
           and scalability through coarse-grained coordination are not all found in current systems. Shared arrangements combine these features in a single system.}
  \label{t:bg-related}
\end{figure*}


Shared arrangements allow queries to share indexed state.
Inter-query state sharing can be framed in terms of \rom{1} \emph{what} can be shared between queries, \rom{2} if this shared state can be \emph{updated}, and \rom{3} the \emph{coordination} required to maintain it.
Figure \ref{t:bg-related} compares sharing in different classes of systems.
%

%
\textbf{Relational databases} like PostgreSQL~\cite{postgresql} excel at
answering queries over schema-defined tables.
Indexes help them speed up access to records in these tables, turning
sequential scans into point lookups.
When the underlying records change, the database updates the index.
This model is flexible and shares indexes between different queries, but it
requires coordination (\eg locking~\cite{sql-server-matview-locks}).
Scaling this coordination out to many parallel processors or servers holding
shards of a large database has proven difficult, and scalable systems consequently
restrict coordination.
Parallel-processing \textbf{``big data'' systems} like
MapReduce~\cite{mapreduce}, Dryad~\cite{dryad}, and Spark~\cite{spark}
rely only on coarse-grained coordination.
They avoid indexes and turn query processing into parallel
scans of distributed collections.
But these collections are immutable: any change to a distributed
collection (e.g. a Spark RDD) requires reconstituting that
collection as a new one.
This captures a collection's lineage and makes all parallelism
deterministic, which eases recovery from failures.
Immutability allows different queries to share the (static)
collection for reading~\cite{nectar}.
This design aids scale-out, but makes these systems a poor fit for
streaming computations, with frequent fine-grained changes to the collections.
\textbf{Stream-processing systems} reintroduce fine-grained mutability, but they lack sharing.
Systems like Flink~\cite{flink}, Naiad~\cite{naiad}, and Noria~\cite{noria} keep long-lived, indexed intermediate results in memory for efficient incremental processing, partitioning the computation across workers for scale-out, data-parallel processing.
However, stream processors associate each piece of state \emph{exclusively} with a single operator, since concurrent accesses to this state from multiple operators would race with state mutations.
Consequently, these systems \emph{duplicate} the state that operators could, in principle, share.
%

By contrast, \textbf{shared arrangements} allow for fine-grained updates to shared indexes and preserve the scalability of data-parallel computation.
In particular, shared arrangements rely on multiversioned indices and data-parallel sharding to allow updates to shared state without the costly coordination mechanisms of classic databases.
In exchange for scalability and parallelism, shared arrangements give up some abilities.
Unlike indexes in relational databases, shared arrangements do not support multiple writers, and are not suitable tools to implement a general transaction processor.
%
Because sharing entangles queries that would otherwise execute in isolation, it can
reduce performance and fault isolation between queries compared to redundant,
duplicated state.
It is important to contrast shared arrangements to Multi-Query Optimization (MQO)
mechanisms that identify overlapping subqueries.
MQO shares state and processing between queries with common subexpressions, but
shared arrangements also benefit distinct queries that access the same indexes.
Both relational and big data systems can identify common sub-expressions via MQO
and either cache their results or fuse their computation.
For example, CJoin~\cite{cjoin} and SharedDB~\cite{shareddb} share
table scans between concurrent ad-hoc queries over large, unindexed tables in
data warehousing contexts, and Nectar~\cite{nectar} does so for
DryadLINQ~\cite{dryadlinq} computations.
More recently AStream~\cite{DBLP:conf/sigmod/KarimovRM19} applied the architecture of SharedDB to windowed streaming computation, and can share among queries the resources applied to future windows.
TelegraphCQ~\cite{telegraphcq} and DBToaster~\cite{dbtoaster} share state among
continuous queries, but sequentially process each query without parallelism
or shared indexes.
Noria~\cite{noria} shares computation between queries over streams, but again lacks shared indexes.
In all these systems, potential sharing must be identified at query deployment time; none provide new queries with access to indexed historical state.
In constrast, shared arrangements (like database indices) allow for post-hoc sharing: new queries can immediately attach to the in-memory arrangements of existing queries, and quickly start producing correct outputs that reflect all prior events.
%

%
%

%
Philosophically closest to shared arrangements is STREAM~\cite{stanford-stream}, a relational stream processor which maintains ``synopses'' (often indexes) for operators and shares them between operators.
In contrast to shared arrangements, STREAM synopses lack features necessary for coarse-grained data-parallel incremental view maintenance: STREAM synopses are not multiversioned and do not support sharding for data-parallelism.
STREAM processes records one-at-a-time; shared arrangements expose a stream of shared, indexed batches to optimized implementations of the operators.

Shared arrangements allow for operators fundamentally designed around shared indexes.
Their ideas are, in principle, compatible with many existing stream processors
that provide versioned updates (as \eg Naiad and Flink do) and support physical co-location of operator shards (as \eg Naiad and Noria do).
%


\section{Context and Overview}
\seclabel{overview}

Shared arrangements are designed in the context of streaming dataflow systems which provide certain core functionality.
We enumerate the requirements in~\subsecref{overview}{timeaware}, and describe how several popular systems meet those requirements in~\subsecref{overview}{timeaware-systems}.
Our implementation builds on Timely Dataflow~\cite{timely}, which offers performant implementations of key abstractions required by shared arrangements~\subsecref{overview}{efficiency}.
With this context, \subsecref{overview}{example} shows how shared arrangements support deployment and continual maintenance of multiple queries against evolving data with the example of TPC-H Q3 and Q5.

\subsection{Time-aware Dataflow}
\subseclabel{overview}{timeaware}

We designed shared arrangements for use in streaming dataflow
computations that implement incrementally maintained queries on high-rate streams.
Data-parallel stream processing systems express such computations as a
dataflow graph whose vertices are \emph{operators}, and whose roots
constitute \emph{inputs} to the dataflow.
An \emph{update} (\eg an event in stream) arrives at an input and flows
along the graph's edges into operators.
Each operator takes the incoming update, processes it, and emits any
resulting derived updates.

\paragraph{Operator State}
In processing the update, a dataflow operator may refer to its \emph{state}: long-lived information that the operator maintains across invocations.
State allows for efficient incremental processing, such as keeping a running counter.
For many operators, the state is indexed by a \emph{key} contained in the input update.
For example, a \texttt{count} operator over tweets grouped by the user who posted them will access its state by user ID.
It is these indexes that shared arrangements seek to share between multiple operators.

\paragraph{Data Parallelism}
Dataflow systems achieve parallel processing by \emph{sharding} operators whose state is indexed by key.
The system partitions the key space, and creates operators to independently process each partition.
In the tweet counting example, the system may partition updates by the user ID, and send each update to an appropriate operator shard, which maintains an index for its subset of user IDs.
Each operator shard maintains its own private index; these index shards, taken collectively, represent the same index a single operator instance would maintain.

\paragraph{Logical Timestamps}
Updates flow through a dataflow graph asynchronously.
Concurrent updates may race along the multiple paths (and even cycles) between dataflow operators potentially distributed across multiple threads of control, and arrive in different orders than they were produced.
For operators to compute correct results in the face of this asynchrony, some coordination mechanism is required.
Many systems assign a logical timestamp to messages, either explicitly or implicitly through their scheduling mechanisms.
At the same time, systems need to inform operators in the dataflow graphs when each logical time has ``passed'', in the sense that it will not again appear on messages input to the operator.
With logical timestamps on messages and timestamp progress statements from the system, operators can maintain clear semantics even with asynchronous, non-deterministic execution.

We use the terminology of Timely Dataflow to describe progress statements and their consequences.
Timely Dataflow reports timestamp progress information to each operator input by a \emph{frontier}: a set of logical timestamps.
We say a time is \emph{beyond} a frontier when it is greater than or equal to some element of the frontier.
A system should guarantee that all future timestamps received at an operator input are beyond the frontier most recently reported by the system, and that these reports should only advance (\ie elements of a frontier should each be beyond the prior frontier).

\subsection{Time-aware Dataflow Systems}
\subseclabel{overview}{timeaware-systems}

Several dataflow systems are time-aware, either implicitly or explicitly. We now give examples to relate the concepts for readers familiar with these systems. Shared arrangements can be implemented in each of these systems, but our implementation will benefit from specific system details, which we explain in ~\subsecref{overview}{efficiency}.

\textbf{Spark Streaming}~\cite{spark-streaming} partitions logical time into small batches, and for each batch evaluates an entire dataflow. Spark Streaming therefore implicitly provides logical timestamps, with progress indicated by the scheduling of an operator. Spark Streaming operators do not have long-lived state, but each invocation can read an input corresponding to its prior state and write an output for its updated state, at greater expense than updating in-memory state.

\textbf{Flink}~\cite{flink} is a streaming dataflow system that timestamps each message, and flows control messages, called \textit{low watermarks}, in-band with data messages. A "watermark" for a timestamp $ t $ indicates that all messages that follow have timestamps greater or equal to $ t $. Flink operators can have long-lived state, and can themselves be the result of sharding a larger dataflow operator.



\textbf{Timely Dataflow} is a model for data-parallel dataflow execution, introduced by Naiad~\cite{naiad}.
Each Timely Dataflow operator is sharded across all workers, with data exchanged between workers for dataflow edges where the destination operator requires it.
In Timely Dataflow, all data carries a logical timestamp, and timestamp progress statements are exchanged out-of-band by workers.
Workers then independently determine frontiers for each of their hosted operators.


%
%

%

\subsection{Shared Arrangements Overview}
\subseclabel{overview}{shared-arrangements-overview}



%
The high-level objective of shared arrangements is to share indexed operator
state, both within a single dataflow and across multiple concurrent dataflows,
serving concurrent continuous queries.
Shared arrangements substitute for per-instance operator state in the dataflow, and  should appear to an individual operator as if it was a private copy of its state.
Across operators, the shared arrangement's semantics are identical to maintaining individual copies of the indexed state in each operator.
At the same time, the shared arrangement permits index reuse between operators
that proceed at a different pace due to asynchrony in the system.
%


%
Operators that provide incremental view maintenance, so that their output
continually reflects their accumulated input updates, offer particularly good
opportunities for sharing state.
This is because each stream of updates has one logical interpretation:
as an accumulation of all updates.
When multiple such operators want to build the same state, but vary what subset to read based on the time $t$ they are currently processing, they can share arrangements instead.
We assume that developers specify their dataflows using existing interfaces,
but that they (or an optimizing compiler) explicitly indicate which dataflow
state to share among which operators.
A shared arrangement exposes different \emph{versions} of the underlying state to
different operators, depending on their current time $t$.
The arrangement therefore emulates, atop physically shared state, the
separate indexes that operators would otherwise keep.
Specifically, shared arrangements maintain state for operators whose state
consists of the input collection (\ie the cumulative streaming input).
Following Differential Dataflow~\cite{differential-dataflow} terminology, a
\emph{collection trace} is the set of update triples \textit{(data, time, diff)}
that define a collection at time $t$ by the accumulation of those inputs
\textit{(data, diff)} for which $\textit{time} \le t$
(Figure~\ref{fig:collection-trace}).
Each downstream operator selects a different view based on a different time
$t$ of accumulation.
Formal semantics of differential dataflow operators are presented in \cite{differential-foundations}.

\begin{figure}

{\renewcommand{\arraystretch}{.8}\setlength{\tabcolsep}{1.5pt}
\begin{tabular}{lll}
Collection trace \\ \hline
{\scriptsize \texttt{(data=(id=342, "Company LLC", "USA"),       }}  & {\scriptsize \texttt{time=4350,}} & {\scriptsize \texttt{diff=+1)}} \\
{\scriptsize \texttt{(data=(id=563, "Firma GmbH", "Deutschland"),}}  & {\scriptsize \texttt{time=4355,}} & {\scriptsize \texttt{diff=+1)}} \\
{\scriptsize \texttt{(data=(id=225, "Azienda SRL", "Italia"),}}      & {\scriptsize \texttt{time=4360,}} & {\scriptsize \texttt{diff=+1)}} \\
{\scriptsize \texttt{(data=(id=225, "Azienda SRL", "Italia"),}}      & {\scriptsize \texttt{time=6200,}} & {\scriptsize \texttt{diff=-1)}} \\
{\scriptsize \texttt{(data=(id=225, "Company Ltd", "UK"),    }}      & {\scriptsize \texttt{time=6220,}} & {\scriptsize \texttt{diff=+1)}} \\ \hline
\end{tabular}
}

\vspace{.2em}

{\renewcommand{\arraystretch}{.8}\setlength{\tabcolsep}{1.5pt}
\begin{tabular}{ll}
Collection at time $t = 4360$ \\ \hline
{\scriptsize \texttt{(data=(id=342, "Company LLC", "USA"),       }} & {\scriptsize \texttt{diff=+1)}} \\
{\scriptsize \texttt{(data=(id=563, "Firma GmbH", "Deutschland"),}} & {\scriptsize \texttt{diff=+1)}} \\
{\scriptsize \texttt{(data=(id=225, "Azienda SRL", "Italia"),    }} & {\scriptsize \texttt{diff=+1)}} \\ \hline
\end{tabular}
}

\vspace{.2em}

{\renewcommand{\arraystretch}{.8}\setlength{\tabcolsep}{1.5pt}
\begin{tabular}{ll}
Collection at time $t = 6230$ \\ \hline
{\scriptsize \texttt{(data=(id=342, "Company LLC", "USA"),  }} & {\scriptsize \texttt{diff=+1)}} \\
{\scriptsize \texttt{(data=(id=563, "Firma", "Deutschland"),}} & {\scriptsize \texttt{diff=+1)}} \\
{\scriptsize \texttt{(data=(id=225, "Company Ltd", "UK"),   }} & {\scriptsize \texttt{diff=+1)}} \\ \hline
\end{tabular}
}
\caption{Update triples incoming to an operator, a ``collection trace'', and the resulting collection view at different times. \label{fig:collection-trace}}
\end{figure}

An explicit, new \code{arrange} operator maintains the multiversioned state and views, while
downstream operators read from their respective views.
The contents of these views vary according to current logical timestamp frontier
at the different operators: for example, a downstream operator's view may not
yet contain updates that the upstream \code{arrange} operator has already added
into the index for a future logical time if the operator has yet to process
them.
%


%
Downstream operators in the same dataflow, and operators in other dataflows
operating in the same logical time domain, can share the arrangement as long as
they use the same key as the arrangement index.
In particular, sharing can extend as far as the next change of key (an
\code{exchange} operator in Differential Dataflow, or a ``shuffle'' in Flink),
an arrangement-unaware operator (\eg \code{map}, which may change the key),
or an operator that explicitly materializes a new collection.

\subsection{Shared Arrangements Example}
\subseclabel{overview}{example}

We illustrate a concrete use of shared arrangements with the example of TPC-H Q3 and Q5.
Recall that in our target setting, analysts author and execute SQL queries against heavily normalized datasets.
%
%
Relations in analytics queries are commonly normalized into ``fact'' and ``dimension'' tables, the former containing foreign keys into the latter.
While new facts (\eg ad impressions, or line items in TPC-H) are continually added, the dimension tables are also updated (for example, when a customer or supplier updates their information).
The dimension tables are excellent candidates for arrangement by primary keys: we expect many uses of these tables to be joins by primary keys, and each time this happens an arrangement can be shared rather than reconstructed.

TPC-H Q3 retrieves the ten unshipped orders with the highest value. 
This is a natural query to maintain, as analysts work to unblock a potential backlog of valuable orders.
The query derives from three relations---\texttt{lineitem}, \texttt{orders}, and \texttt{customer}---joined using the primary keys on \texttt{orders} and \texttt{customer}.
A dataflow would start from \texttt{lineitem} and join against \texttt{orders} and \texttt{customer} in sequence. 
TPC-H Q5 lists the revenue volume done through local suppliers, and derives from three more relations (\texttt{supplier}, \texttt{nation}, and \texttt{region}).
Each relation other than \texttt{lineitem} is joined using its primary key.
A dataflow might start from \texttt{lineitem} and join against dimension tables in a sequence that makes a foreign key available for each table before joining it.
In both queries, each dimension table is sharded across workers by their primary key.

The two queries do not have overlapping subqueries---each has different filters on order dates, for example---but both join against \texttt{orders} and \texttt{customer} by their primary keys.
Deployed on the same workers, we first apply \texttt{arrange} operators to the \texttt{orders} and \texttt{customer} relations by their primary keys, shuffling updates to these relations by their key and resulting in shareable arrangements.
In separate dataflows, Q3 and Q5 both have \texttt{join} operators that take as input the corresponding arrangement, rather than the streams of updates that formed them.
As each arrangement is pre-sharded by its key, each worker has only to connect its part of each arrangement to its dataflow operators.
Each worker must still stream in the \texttt{lineitem} data but the time for the query to return results becomes independent of the sizes of \texttt{orders} and \texttt{customer}.

%
%

\subsection{System Features Supporting Efficiency}
\subseclabel{overview}{efficiency}

Shared arrangements apply in the general dataflow setting described in~\subsecref{overview}{timeaware}, and can benefit any system with those properties.
But additional system properties can make an implementation more performant.
We base our implementation on frameworks (Timely and Differential Dataflow) with these properties.

\smallparagraph{Timestamp batches}
Timestamps in Timely Dataflow only need to be \emph{partially ordered}.
The partial order of these timestamps allows Timely Dataflow graphs to avoid unintentional concurrency blockers, like serializing the execution of rounds of input (Spark) or rounds of iteration (Flink).
By removing these logical concurrency blockers, the system can retire larger groups of logical times at once, and produce larger batches of updates.
This benefits \sys because the atoms of shared state can increase in granularity, and the coordination between the sharing sites can decrease substantially.
Systems that must retire smaller batches of timestamps must coordinate more frequently, which can limit their update rates.

\smallparagraph{Multiversioned state}
Differential Dataflow has native support for \emph{multiversioned} state.
This allows it to work concurrently on any updates that are not yet beyond the Timely Dataflow frontier, without imposing a serial execution order on updates.
Multiversioned state benefits shared arrangements because it decouples the execution of the operators that share the state.
Without multiversioned state, operators that share state must have their executions interleaved for each logical time, which increases coordination.

\smallparagraph{Co-scheduling}
Timely Dataflow allows each worker to host an unbounded number of dataflow operators, which the worker then schedules.
This increases the complexity of each worker compared to a system with one thread per dataflow operator, but it increases the efficiency in complex dataflows with more operators than system threads.
Co-scheduling benefits shared arrangements because the state shared between operators can be partitioned between worker threads, who do not need mutexes or locks to manage concurrency.
Systems that cannot co-schedule operators that share state must use inter-thread or inter-process mechanisms to access shared state, increasing complexity and the cost.

\smallparagraph{Incremental Updates}
Differential dataflow operators are designed to provide incremental view maintenance: their output updates continually reflects their accumulated input updates.
This restriction from general-purpose stream processing makes it easier to compose dataflows based on operators with clear sharing semantics.
%
%
Systems that provide more general interfaces, including Timely Dataflow, push a substantial burden on to the user to identify operators that can share semantically equivalent state.

\section{Implementation}
\seclabel{arrange}

Our implementation of a shared arrangement consists of three inter-related components:

\begin{enumerate}[nosep]
  \item the \emph{trace}, a list of immutable, indexed batches of updates that together make up the multiversioned index;
  \item an \code{arrange} operator, which mints new batches of updates, and writes them to and maintains the trace; and
  \item read \emph{handles}, through which arrangement-aware operators access the trace.
\end{enumerate}

Each shared arrangement has its updates partitioned by the key of its index, across the participating dataflow workers.
This same partitioning applies to the trace, the \code{arrange} operator, and the read handles, each of whose interactions are purely \emph{intra}-worker; each worker maintains and shares its \emph{shard} of the whole arrangement.
The only inter-worker interaction is the pre-shuffling of inbound updates which effects the partition.

%
%

Figure~\ref{fig:arrange} depicts a dataflow which uses an arrangement for the \texttt{count} operator, which must take a stream of \textit{(data, time, diff)} updates and report the changes to accumulated counts for each \textit{data}. This operation can be implemented by first partitioning the stream among workers by \textit{data}, after which each worker maintains an index from \textit{data} to its history, a list of \textit{(time, diff)}. This same indexed representation is what is needed by the \texttt{distinct} operator, in a second dataflow, which can re-use the same partitioned and indexed arrangement rather than re-construct the arrangement itself.


\begin{figure}
\centering
\includegraphics[width=0.45\textwidth]{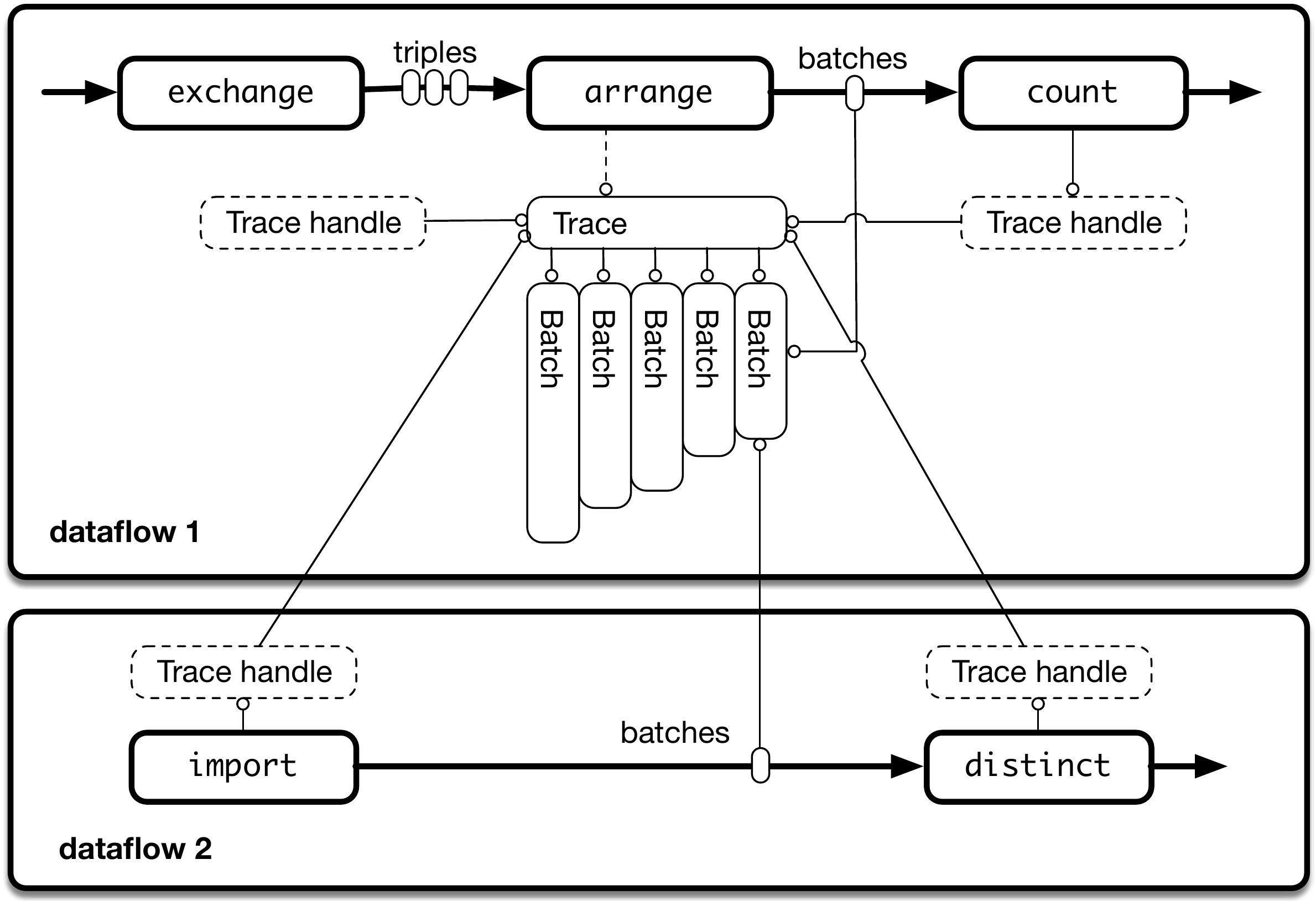}
\caption{A worker-local overview of arrangement. Here the arrangement is constructed for the \texttt{count} operator, but is shared with a \texttt{distinct} operator in another dataflow. Each other worker performs the same collective data exchange, followed by local batch creation, trace maintenance, and sharing.
\label{fig:arrange}}
\end{figure}

\subsection{Collection traces}

As in Differential Dataflow, a \emph{collection trace} is the set of update triples \textit{(data, time, diff)} that define a collection at any time $t$ by the accumulation of those \textit{(data, diff)} for which $\textit{time} \le t$. A collection trace is initially empty and is only revealed as a computation proceeds, determined either as an input to the dataflow or from the output of another dataflow operator. Although update triples may arrive continually, it is only when the Timely Dataflow input frontier advances that the \arrange operator learns that the updates for a subset of times are now complete.

In our design a collection trace is logically equivalent to an append-only list of immutable batches of update triples. Each batch is described by two frontiers of times, \emph{lower} and \emph{upper}, and the batch contains exactly those updates whose times are beyond the lower frontier and not beyond of the upper frontier. The upper frontier of each batch matches the lower frontier of the next batch, and the growing list of batches reports the developing history of confirmed updates triples. A batch may be empty, which indicates that no updates exist in the indicated range of times.

To support efficient navigation of the collection trace, each batch is
indexed by its \textit{data} to provide random access to the history
of each \textit{data} (the set of its $(\textit{time}, \textit{diff})$
pairs).
Background merge computation (performed by the \code{arrange} operator) ensures that at any time, a trace consists of logarithmically many bat\-ches, which ensures that operators can efficiently navigate the union of all batches.

Each reader of a trace holds a \emph{trace handle}, which acts as a cursor that can navigate the multiversioned index. Each handle has an associated frontier, and ensures that it provides correct views of the index for any times beyond this frontier.
Trace readers advance the frontier of their trace handle when they no longer require certain historical distinctions, which allows the \arrange operator to compact batches by coalescing updates at older times, and to maintain a bounded memory footprint as a collection evolves.



\subsection{The \texttt{arrange} operator}
\subseclabel{arrangements}{arrange}

The \arrange operator receives update triples, and must both create new immutable indexed batches of updates as its input frontier advances and compactly maintain the collection trace without violating its obligations to readers of the trace.

At a high level, the \arrange operator buffers incoming update triples until the input frontier advances, at which point it extracts and indexes all buffered updates not beyond the newly advanced input frontier. A shared reference to this new immutable batch is both added to the trace and emitted as output from the \arrange operator. When adding the batch to the trace, the operator may need to perform some maintenance to keep the trace representation compact and easy to navigate.



\paragraph{Batch implementation}
Each batch is immutable, but indexed to provide efficient random access.
Our default implementation sorts update triples \textit{(data, time, diff)} first by \textit{data} and then by \textit{time}, and stores the fields each in its own column.
This balances the performance of read latency, read throughput, and merge throughput.
We have other batch implementations for specific domains (\eg graphs), and new user implementations can be added without changing the surrounding superstructure.
Most OLTP index structures are more general than needed for our immutable batches, but many of their data layout ideas could still be imported.

\paragraph{Amortized trace maintenance} The maintenance work of merging batches in a trace is amortized over the introduced batches, so that no batch causes a spike in computation (and a resulting spike in latency). Informally, the operator performs the same \emph{set} of merges as would a merge sort applied to the full sequence of batches, but only as the batches become available. Additionally, each merge is processed in steps: for each new batch, we perform work proportional to the batch size on each incomplete merge. A higher constant of proportionality leads to more eager merging, improving the throughput of the computation, whereas a lower constant improves the maximum latency of the computation.

\paragraph{Consolidation} \subseclabel{arrange}{consolidation} As readers of the trace advance through time, historical times become indistinguishable and updates at such times to the same \textit{data} can be coalesced. The logic to determine which times are indistinguishable is present in Naiad's prototype implementation~\cite{naiad}, but the mathematics of compaction have not been reported previously.
\extendedtr{%
In \ref{appendix:theorems}, we present proofs of optimality and correctness.
}{%
Our \theTR contains proofs of optimality and correctness.
}

\paragraph{Shared references} Both immutable batches and traces themselves are reference counted. Importantly, the \arrange operator holds only a ``weak'' reference to its trace, and if all readers of the trace drop their handles the operator will continue to produce batches but cease updating the trace. This optimization is crucial for competitive performance in computations that use both dynamic and static collections.

\subsection{Trace handles}

Read access to a collection trace is provided through a \emph{trace
  handle}. A trace handle provides the ability to \texttt{import} a
collection into a new dataflow, and to manually navigate a collection,
but both only ``as of'' a restricted set of times. Each trace handle
maintains a frontier, and guarantees only that accumulated collections
will be correct when accumulated to a time beyond this
frontier. The trace itself tracks outstanding trace handle frontiers,
which indirectly inform it about times that are indistinguishable to
all readers (and which can therefore be coalesced).

Many operators (including \texttt{join} and \texttt{group}) only need access to their accumulated input collections for times beyond their input frontiers. As these frontiers advance, the operators are able to advance the frontier on their trace handles and still function correctly. The \texttt{join} operator is even able to drop the trace handle for an input when its \emph{other} input ceases changing. These actions, advancing the frontier and dropping trace handles, provide the \arrange operator with the opportunity to consolidate the representation of its trace, and in extreme cases discard it entirely.

A trace handle has an \texttt{import} method that, in a new dataflow, creates an
arrangement exactly mirroring that of the trace.
The imported collection immediately produces any existing consolidated historical
batches, and begins to produce newly minted batches.
The historical batches reflect all updates applied to the collection, either with
full historical detail or coalesced to a more recent timestamp, depending on whether
the handle's frontier has been advanced before importing the trace.
Computations require no special logic or modes to accommodate attaching to
in-progress streams; imported traces appear indistinguishable to their
original streams, other than their unusually large batch sizes and recent
timestamps.


\section{Arrangement-aware operators}
\seclabel{operators}

Operators act on collections, which can be represented either as a stream of update triples or as an arrangement. These two representations lead to different operator implementations, where the arrangement-based implementations can be substantially more efficient than traditional record-at-a-time operator implementations. In this section we explain arrangement-aware operator designs, starting with the simplest examples and proceeding to the more complex \texttt{join}, \texttt{group}, and \texttt{iterate} operators.

\subsection{Key-preserving stateless operators}
\subseclabel{operators}{stateless}

Several stateless operators are ``key-preserving'': they do not transform their input data to the point that it needs to be re-arranged. Example operators are \texttt{filter}, \texttt{concat}, \texttt{negate}, and the iteration helper methods \texttt{enter} and \texttt{leave}. These operators are implemented as streaming operators for streams of update triples, and as wrappers around arrangements that produce new arrangements. For example, the \texttt{filter} operator results in an arrangement that applies a supplied predicate as part of navigating through a wrapped inner arrangement.
This design implies a trade-off, as an aggressive filter may
reduce the data volume to the point that it is cheap to
maintain a separate index, and relatively ineffective to search in a
large index only to discard the majority of results.
The user controls which implementation to use: they can filter an
arrangement, or reduce the arrangement to a stream of updates and then filter it.


\subsection{Key-altering stateless operators}
\subseclabel{operators}{key-altering}

Some stateless operators are ``key-altering'', in that the indexed representation of their output has little in common with that of their input. One obvious example is the \texttt{map} operator, which may perform arbitrary record-to-record transformations. These operators always produce outputs represented as streams of update triples.

\subsection{Stateful operators}
\subseclabel{operators}{stateful}

Differential Dataflow's stateful operators are data-parallel: their input \textit{data} have a \textit{(key, val)} structure, and the computation acts independently on each group of \textit{key} data. This independence is what allows Naiad and similar systems to distribute operator work across otherwise independent workers, which can then process their work without further coordination. At a finer scale, this independence means that each worker can determine the effects of a sequence of updates on a key-by-key basis, resolving all updates to one key before moving to the next, even if this violates timestamp order.



\subsubsection{The \texttt{join} operator}
\subseclabel{operators}{op-join}

Our \texttt{join} operator takes as inputs batches of updates from each of its arranged inputs. It produces any changes in outputs that result from its advancing inputs, but our implementation has several variations from a traditional streaming hash-join.

\paragraph{Trace capabilities}
The join operator is bi-linear, and needs only each input trace in order to respond to updates from the \emph{other} input. As such, the operator can advance the frontiers of each trace handle by the frontier of the other input, and it can drop each trace handle when the other input closes out. This is helpful if one input is static, as in iterative processing of static graphs.

\paragraph{Alternating seeks}
Join can receive input batches of substantial size, especially when
importing an existing shared arrangement. Naively implemented, we might require time linear in the input batch sizes. Instead, we perform alternating seeks between the cursors for input batches and traces of the other input: when the cursor keys match we perform work, and if the keys do not match we seek forward for the larger key in the cursor with the smaller key. This pattern ensures that we perform work at most linear in the smaller of the two sizes, seeking rather than scanning through the cursor of the larger trace, even when it is supplied as an input batch.

\paragraph{Amortized work}
The \texttt{join} operator may produce a significant amount of output data that can be
reduced only once it crosses an exchange edge for a downstream operator.
If each input batch is immediately processed to completion, workers may be overwhelmed
by the output, either buffered for transmission or (as in our prototype) sent to
destination workers but buffered at each awaiting reduction.
Instead, operators respond to input batches by producing ``futures'', limited
batches of computation that can each be executed until sufficiently many outputs
are produced, and then suspend.
Futures make copies of the shared batch and trace references they use, which avoids
blocking state maintenance for other operators.

\subsubsection{The \texttt{group} operator}
\subseclabel{operators}{op-group}

The \texttt{group} operator takes as input an arranged collection with data of the form \textit{(key, val)} and a reduction function from a key and list of values to a list of output values. At each time the output might change, we reform the input and apply the reduction function, and compare the results to the reformed output to determine if output changes are required.

Perhaps surprisingly, the output may change at times that do not appear in the input (as the least upper bound of two times does not need to be one of the times). Hence, the \texttt{group} operator tracks a list of pairs \textit{(key, time)} of future work that are required even if we see no input updates for the key at that time. For each such \textit{(key, time)} pair, the \texttt{group} operator accumulates the input and output for \textit{key} at \textit{time}, applies the reduction function to the input, and subtracts the accumulated output to produce any corrective output updates.

\paragraph{Output arrangements}
The \texttt{group} operator uses a shared arrangement for its output, to efficiently reconstruct what it has previously produced as output without extensive re-invocation of the supplied user logic (and to avoid potential non-determinism therein). This provides the \texttt{group} operator the opportunity to share its output trace, just as the \texttt{arrange} operator does. It is common, especially in graph processing, for the results of a \texttt{group} to be immediately joined on the same key,
and \texttt{join} can re-use the same indexed representation that \texttt{group} uses internally for its output.


\subsection{Iteration}
\subseclabel{operators}{iteration}

The iteration operator is essentially unchanged from Naiad's Differential Dataflow implementation. We have ensured that arrangements can be brought in to iterative scopes from outer scopes using only an arrangement wrapper, which allows access to shared arrangements in iterative computations.




\section{Evaluation}
\seclabel{eval}

%
%
We evaluate \sys on end-to-end workloads to measure the impact of shared arrangements with regards to query installation latency, throughput, and memory use (\subsecref{eval}{endtoend}).
We then use microbenchmarks with \sys to characterize our design's performance and the arrangement-aware operator implementations (\subsecref{eval}{designeval}).
Finally, we evaluate \sys on pre-existing benchmarks across multiple domains to check if \sys maintains high performance compared to other peer systems with and without using shared arrangements (\subsecref{eval}{broad}).
%



\paragraph{Implementation}
\seclabel{impl}
We implemented shared arrangements as part of \sys, our stream processor.
\sys is our reference Rust implementation of Differential Dataflow~\cite{differential-dataflow}
with shared arrangements. It consists of a total of about 11,700 lines of code, and
builds on an open-source implementation of Timely Dataflow~\cite{timely}.
%

%
%
%
%

%
The \texttt{arrange} operator is defined in terms of a generic trace type, and
our amortized merging trace is defined in terms of a generic batch type.
Rust's static typing ensure that developers cannot incorrectly mix ordinary
update triples and streams of arranged batches.
%
%
%
%

\paragraph{Setup}
We evaluate \sys on a four-socket NUMA system with four Intel Xeon
E5-4650 v2 CPUs, each with 10 physical cores and 512 GB of aggregate system
memory.
We compiled \sys with \texttt{rustc} 1.33.0 and the \texttt{jemalloc}~\cite{jemalloc} allocator.
\sys does distribute across multiple machines and supports sharding shared arrangements across them, but our evaluation here is restricted to multiprocessors.
%
%
When we compare against other systems, we rely on the best, tuned measurements reported by their authors, but compare \sys only if we are executing it on comparable or less powerful hardware than the other systems had access to.
%



\subsection{End-to-end impact of shared arrangements}
\subseclabel{eval}{endtoend}

We start with an evaluation of shared arrangements in \sys, in two domains with interactively issued queries against incrementally updated data sources.
We evaluate the previously described streaming TPC-H setup, which windows the \texttt{lineitem} relation, as well as a recent interactive graph analytics benchmark.
For the relational queries, we would hope to see shared arrangements reduce the installation latency and memory footprint of new queries when compared to an instance of \sys that processes queries independently.
For the graph tasks, we would hope that shared arrangements reduce the update and query latencies at each offered update rate, increase the peak update rate, and reduce the memory footprint when compared to an instance of \sys that processes queries indepedently.
In both cases, if shared arrangements work as designed, they should increase the capacity of \sys on fixed resources, reducing the incremental costs of new queries.

\subsubsection{TPC-H}
\subseclabel{eval}{tpch}
The TPC-H~\cite{tpch} benchmark schema has eight relations, which describe order fulfillment events, as well as the orders, parts, customers, and suppliers they involve, and the nations and regions in which these entities exist.
Of the eight relations, seven have meaningful primary keys, and are immediately suitable for arrangement (by their primary key).
The eighth relation is \texttt{lineitem}, which contains fulfillment events, and we treat this collection as a stream of instantaneous events and do not arrange it.

TPC-H contains 22 ``data warehousing'' queries, meant to be run against large, static datasets.
We consider a modified setup where the eight relations are progressively loaded~\cite{hotdog},
one record at a time, in a round-robin fashion among the relations (at scale factor
10).\footnote{We focus on shared arrangements here, but \sys matches or
outperforms DBToaster~\cite{hotdog} even when queries run in isolation
\extendedtr{%
We show these results in \ref{appendix:relational}.
}{%
\cite{kpg-tr}.
}
}
To benchmark the impact of shared arrangements,
we interactively deploy and retire queries
while we load the eight relations.
Each query has access to the full current contents of the seven keyed relations
that we maintain shared arrangements for.
By contrast, fulfillment events are windowed and each query only observes the
fulfillment events from when it is deployed until when it is retired,
implementing a ``streaming'' rather than a ``historic'' query.
This evaluates the scenario presented in~\secref{intro}, where analysts
interactively submit queries.
We report performance for ten active queries.
The 22 TPC-H queries differ, but broadly either derive from the windowed \code{lineitem}
relation and reflect only current fulfillments, or they do not derive from \code{lineitem}
and reflect the full accumulated volume of other relations.
Without shared arrangements, either type of query requires building new
indexed state for the seven non-\code{lineitem} relations.
With shared indexes, we expect queries of the first type to be quick
to deploy, as their outputs are initially empty.
Queries of the second type should take longer to deploy
in either case, as their initial output depends on many records.

\paragraph{Query latency}
To evaluate query latency, we measure the time from the start of query
deployment until the initial result is ready to be returned.
Query latency is significant because it determines whether the system delivers an interactive experience to human users, but also to dashboards that programmatically issue queries.
Figure \ref{fig:eval:tpch-hist-install} (shown in \S\ref{sec:intro}) reports the distribution
of query installation latencies, with and without shared arrangements.
With shared arrangements, most queries (those that derive from
\code{lineitem}) deploy and begin updating in milliseconds; the five
queries that do not derive from \texttt{lineitem} are not windowed and perform non-trivial
computation to produce their initial correct answer: they take between 100ms
and 1s, depending on the sizes of the relations they use.
Without shared arrangements, almost all queries take 1--2 seconds to
install as they must create a reindexed copy of their inputs.
Q1 and Q6 are exceptions, since they use no relations other than
\texttt{lineitem}, and thus avoid reindexing any inputs; shared arrangements cannot improve the installation latency of these queries.
We conclude that shared arrangements substantially reduce the majority of query installation latencies, often by several orders of magnitude.
The improvement to millisecond latency brings responses within interactive timescales,
which helps improve productivity of human analysts and intervential latency for dependent
software.

\paragraph{Update latency}
Once a query is installed, \sys continually updates its results as new \code{lineitem} records arrive.
To evaluate the update latency achieved, we record the amount of time required to process each round of input data updates after query installation.
%
%

%
Figure \ref{fig:eval:tpch-hist-work} presents the distribution of these times,
with and without shared arrangements, as a complementary cumulative
distribution function (CCDF).
The CCDF visualization---which we will use re\-peat\-ed\-ly---shows the ``fraction of
times with latency greater than'' and highlights the tail latencies towards
the bottom-right side of the plot.
We see a modest but consistent reduction in processing time (about 2$\times$)
when using shared arrangements, which eliminate redundant index maintenance
work.
There is a noticeable tail in both cases, owed to two expensive queries that
involve inequality joins (Q11 and Q22) and which respond slowly to changes
in their inputs independently of shared arrangements.
Shared arrangements yield lower latencies and increase update throughput.

\begin{figure*}[t]
\begin{subfigure}[b]{0.329\textwidth}
\includegraphics[trim={0 0 0 0}, clip, width=\textwidth]{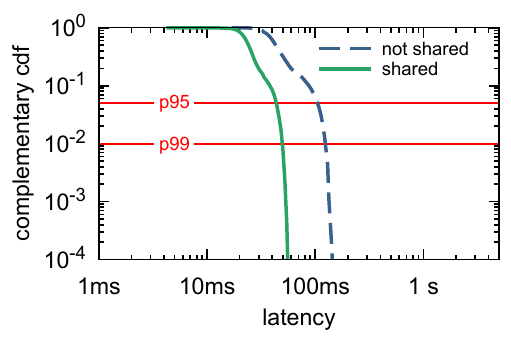}
\caption{Latencies for query mix.}
\label{fig:eval:graph_sharing}
\end{subfigure}
~
\begin{subfigure}[b]{0.329\textwidth}
\includegraphics[trim={0 0 0 0}, clip, width=\textwidth]{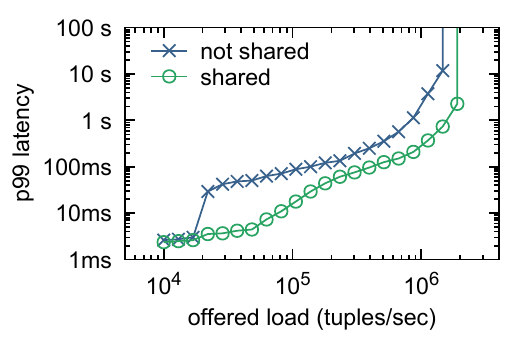}
\caption{99\textsuperscript{th} percentile latency at given offered load.}
\label{fig:eval:graph_boomerang}
\end{subfigure}
~
\begin{subfigure}[b]{0.329\textwidth}
\includegraphics[trim={0 0 0 0}, clip, width=\textwidth]{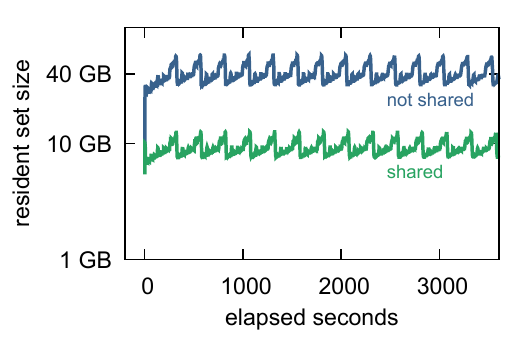}
\caption{Resident set size.}
\label{fig:eval:graph_memory}
\end{subfigure}
  \caption{Shared arrangements reduce query latency, increase the load handled, and reduce the memory footprint of interactive graph queries. The setup uses 32 workers, and issues 100k updates/sec and 100k queries/sec against a 10M node/64M edge graph in (\subref{fig:eval:graph_sharing}) and (\subref{fig:eval:graph_memory}), while (\subref{fig:eval:graph_boomerang}) varies the load. Note the $\log_{10}$--$\log_{10}$ scales in (\subref{fig:eval:graph_sharing}) and (\subref{fig:eval:graph_boomerang}), and the $\log_{10}$-scale $y$-axis in (\subref{fig:eval:graph_memory}).}
\label{fig:eval:graph_maintain}
\end{figure*}

\paragraph{Memory footprint}
Since shared arrangements eliminate duplicate copies of index structures, we
would expect them to reduce the dataflow's memory footprint.
To evaluate the memory footprint, we record the resident set size (RSS) as
the experiment proceeds.
Figure~\ref{fig:eval:tpch-rss} presents the timelines of the RSS with and
without shared arrangements, and shows a substantial reduction (2--3$\times$)
in memory footprint when shared arrangements are present.
Without shared arrangements, the memory footprint also varies substantially
(between 60 and 120 GB) as the system creates and destroys indexes for queries
that arrive and depart, while shared arrangements remain below 40 GB.
Consequently, with shared arrangements, a given amount of system memory should
allow for more active queries.
In this experiment, ten concurrent queries are installed; workloads
with more concurrent queries may have more sharing opportunities and achieve
further memory economies.

\subsubsection{Interactive graph queries}
\subseclabel{eval}{interactive}

%
We further evaluate \sys with an open-loop experiment issuing queries against
an evolving graph.
This experiment issues the four queries used by Pacaci et
al.~\cite{Pacaci:2017:WNS:3078447.3078459} to compare relational and graph
databases: point look-ups, 1-hop look-ups, 2-hop look-ups, and 4-hop shortest
path queries (shortest paths of length at most four).
In the first three cases, the query argument is a graph node identifier, and in
the fourth case it is a pair of identifiers.
We implement each of these queries as Differential Dataflows where the query
arguments are independent collections that may be modified to introduce or
remove specific query arguments.
This query transformation was introduced in NiagaraCQ~\cite{niagaracq} and is
common in stream processors, allowing them to treat queries as a streaming
input.
The transformation can be applied to any queries presented as prepared statements.
The dataflows depend on two arrangements of the graph edges, by source and by target; they
are the only shared state among the queries.

We use a graph with 10M nodes and 64M edges, and update the graph and query
arguments of interest at experiment-specific rates.
Each graph update is the addition or removal of a random graph edge, and each
query update is the addition or removal of a random query argument (queries
are maintained while installed, rather than issued only once).
All experiments evenly divide the query updates between the four query types.

\begin{figure}[h]
  \footnotesize
  \begin{center}
    \begin{tabular}{l|r|r|r|r|r}
      \textbf{System}      & \#    & look-up & one-hop & two-hop & 4-path     \\ \hline
      Neo4j                & 32    & 9.08ms  & 12.82ms & 368ms   & 21ms       \\
      Postgres             & 32    & 0.25ms  & 1.4ms   & 29ms    & 2242ms     \\
      Virtuoso             & 32    & 0.35ms  & 1.23ms  & 11.55ms & 4.81ms     \\ \hline
      \sys, $10^0$         & 32    & 0.64ms  & 0.92ms  & 1.28ms  & 1.89ms     \\
      \sys, $10^1$         & 32    & 0.81ms  & 1.19ms  & 1.65ms  & 2.79ms     \\
      \sys, $10^2$         & 32    & 1.26ms  & 1.79ms  & 2.92ms  & 8.01ms     \\
      \sys, $10^3$         & 32    & 5.71ms  & 6.88ms  & 10.14ms & 72.20ms    \\
      \end{tabular}
    \caption{On comparable 10M node/64M edge graphs, \sys is broadly competitive with the average graph query latencies of three systems evaluated by Pacaci et al.~\protect\cite{Pacaci:2017:WNS:3078447.3078459}, and scales to higher throughput using batching. The \sys batch size is the number of concurrent queries per measurement.}
    \label{tab:graphs_interactive}
  \end{center}
\end{figure}

\paragraph{Query latency}
We run an experiment with a fixed rate of 100,000 query updates
per second, independently of how quickly \sys responds to them.
We would hope that \sys responds quickly, and that shared arrangements of the
graph structure should help reduce the latency of query updates, as \sys must
apply changes to one shared index rather than several independent ones.
Figure~\ref{fig:eval:graph_sharing} reports the latency distributions with
and without a shared arrangement of the graph structure, as a complementary
CDF.
Sharing the graph structure results in a 2--3$\times$ reduction in overall
latency in the 95\textsuperscript{th} and 99\textsuperscript{th} percentile
tail latency (from about 150ms to about 50ms).
In both cases, there is a consistent baseline latency, proportional
to the number of query classes maintained.
Shared arrangements yield latency reductions across all query classes, rather
than, \eg imposing the latency of the slowest query on all sharing
dataflows.
This validates that queries can proceed at different rates, an important
property of our shared arrangement design.

\paragraph{Update throughput}
To test how \sys's shared arrangements scale with load, we next scale
the rates of graph updates and query changes up to two million changes per
second each.
An ideal result would show that sharing the arranged graph structure
consistently reduces the computation required, thus allowing us to scale to a
higher load using fixed resources.
Figure~\ref{fig:eval:graph_boomerang} reports the 99\textsuperscript{th}
percentile latency with and without a shared graph arrangement, as a function of
offered load and on a $\log$--$\log$ scale.
The shared configuration results in reduced latencies at each offered load, and
tolerates an increased maximum load at any target latency.
At the point of saturating the server resources, shared arrangements tolerate
33\% more load than the unshared setup, although this number is much larger for
specific latencies (\eg 5$\times$ at a 20ms target).
We note that the absolute throughputs achieved in this experiment exceed the
best throughput observed by Pacaci et al.\ (Postgres, at 2,000 updates per
second) by several orders of magnitude, further illustrating the benefits of
parallel dataflow computation with shared arrangements.

\paragraph{Memory footprint}
Finally, we consider the memory footprint of the computation.
There are five uses of the graph across the four queries, but also per-query
state that is unshared, so we would expect a reduction in memory
footprint of somewhat below 4$\times$.
Figure~\ref{fig:eval:graph_memory} reports the memory footprint for the query
mix with and without sharing, for an hour-long execution.
The memory footprint oscillates around 10 GB with shared arrangements, and
around 40 GB (4$\times$ larger) without shared arrangements.
This illustrates that sharing state affords memory savings proportional to the
number of reuses of a collection.
%


\subsubsection{Comparison with other systems}
\subseclabel{eval}{comparison}

Pacaci et al.~\cite{Pacaci:2017:WNS:3078447.3078459} evaluated relational and
graph databases on the same graph queries.
\sys is a stream processor rather than a database and supports somewhat
different features, but its performance ought to be comparable to the
databases' for these queries.
We stress, however, that our implementation of the queries as Differential Dataflows requires that queries be expressed as prepared statements, a restriction the other systems do not impose.
We ran \sys experiments with a random graph comparable to the one used in Pacaci
et al.'s comparison.
\tableref{tab:graphs_interactive} reports the average latency to perform and
then await a single query in different systems, as well as the time to perform
and await batches of increasing numbers of concurrent queries for \sys.
While \sys does not provide the lowest latency for point look-ups, it does
provides excellent latencies for other queries and increases query throughput
with batch size.

\subsection{Design evaluation}
\subseclabel{eval}{designeval}

\begin{figure*}[t]
\begin{subfigure}[b]{0.329\textwidth}
\includegraphics[trim={0 0.15cm 0 0}, clip, width=\textwidth]{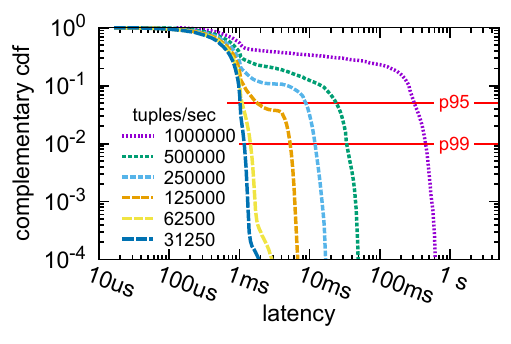}
\caption{Varying offered load with 1 worker.}
\label{fig:eval:load_varies}
\end{subfigure}
~
\begin{subfigure}[b]{0.329\textwidth}
\includegraphics[trim={0 0.15cm 0 0}, clip, width=\textwidth]{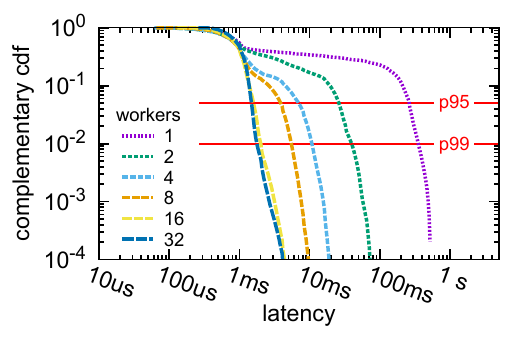}
\caption{Varying workers with fixed load.}
\label{fig:eval:strong_scaling}
\end{subfigure}
~
\begin{subfigure}[b]{0.329\textwidth}
\includegraphics[trim={0 0.15cm 0 0}, clip, width=\textwidth]{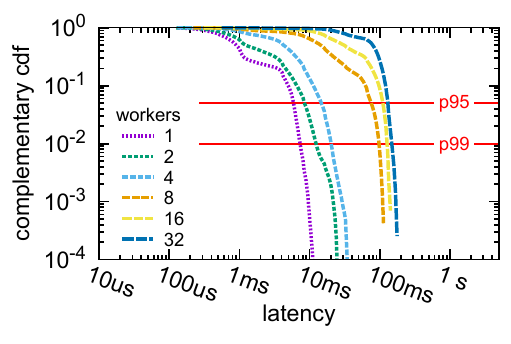}
  \caption{Varying workers and offered load}
\label{fig:eval:weak_scaling}
\end{subfigure}\\

\begin{subfigure}[b]{0.329\textwidth}
\includegraphics[trim={0 0.15cm 0 0}, clip, width=\textwidth]{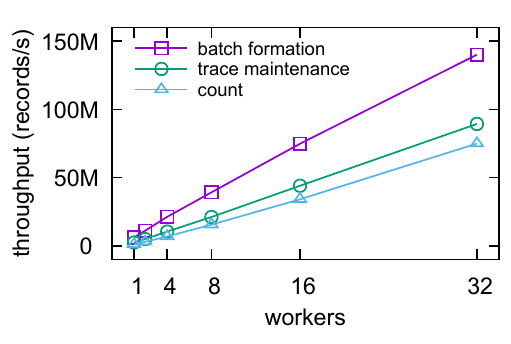}
\caption{Task throughput, varying workers.}
\label{fig:eval:throughput}
\end{subfigure}
~
\begin{subfigure}[b]{0.329\textwidth}
\includegraphics[trim={0 0.15cm 0 0}, clip, width=\textwidth]{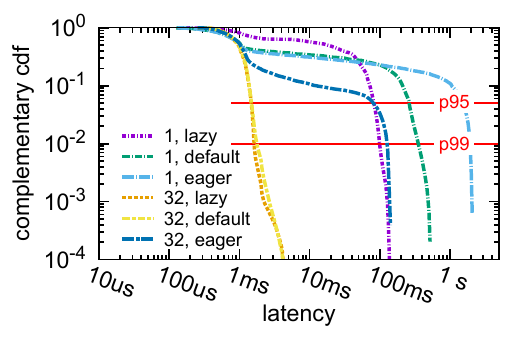}
\caption{Amortized merging levels.}
\label{fig:eval:amortized}
\end{subfigure}
~
\begin{subfigure}[b]{0.329\textwidth}
\includegraphics[trim={0 0.15cm 0 0}, clip, width=\textwidth]{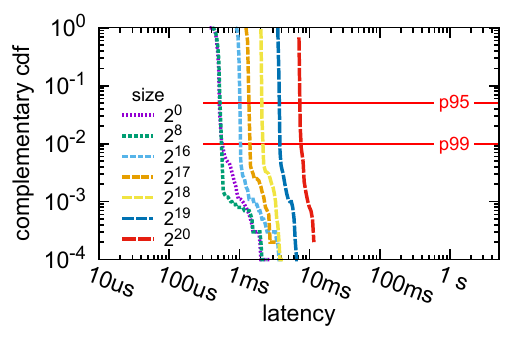}
\caption{Join with pre-arranged collection.}
\label{fig:eval:join_proportionality}
\end{subfigure}
  \caption{Microbenchmarks of our shared arrangement design suggest that our design scales well with growing parallelism ((\subref{fig:eval:strong_scaling})--(\subref{fig:eval:throughput})) and load ((\subref{fig:eval:load_varies}), (\subref{fig:eval:weak_scaling})--(\subref{fig:eval:throughput})), and that the key ideas of amortized merging ((\subref{fig:eval:amortized})) and proportional work across inputs ((\subref{fig:eval:join_proportionality})) are crucial to achieving low update latencies. (\subref{fig:eval:strong_scaling}) and (\subref{fig:eval:amortized}) generate a fixed load of 1M input records per second.}


\end{figure*}

We now perform microbenchmarks of the \texttt{arrange} operator, to evaluate its response to changes in load and resources.
In all benchmarks, we apply an \texttt{arrange} operator to a continually
changing collection of 64-bit identifiers (with 64-bit timestamp and signed
difference).
The inputs are generated randomly at the worker, and exchanged (shuffled)
by key prior to entering the arrangement.
We are primarily interested in the distribution of response latencies, as slow
edge-case behavior of an arrangement would affect this statistic most.
We report all latencies as complementary CDFs to get high resolution in the
tail of the distribution.

\paragraph{Varying load}
As update load varies, our shared arrangement design should trade latency for
throughput until equilibrium is reached.
Figure~\ref{fig:eval:load_varies} reports the latency distributions for a
single worker as we vary the number of keys and offered load in an open-loop
harness, from 10M keys and 1M updates per second, downward by factors of two.
Latencies drop as load decreases, down to the test harness's limit of one millisecond.
This demonstrates that arrangements are suitable for both low-latency and high-throughput.

\paragraph{Strong scaling}
More parallel workers should allow faster maintenance of a shared arrangement,
as the work to update it parallelizes, unless coordination frequency interferes.
Figure~\ref{fig:eval:strong_scaling} reports the latency distributions for
an increasing numbers of workers under a fixed load of 10M keys and 1M updates
per second.
As the number of workers increases, latencies decrease, especially in the tail
of the distribution: for example, the 99\textsuperscript{th} percentile latency
of 500ms with one worker drops to 6ms with eight workers.

\paragraph{Weak scaling}
%
%
Varying the number of workers while proportionately increasing the number of
keys and offered load would ideally result in constant latency.
Figure~\ref{fig:eval:weak_scaling} shows that the latency distributions do
exhibit increased tail latency, as the act of data exchange at the arrangement input
becomes more complex.
However, the latencies do stabilize at 100--200ms as the number of workers and data increase proportionately.

\paragraph{Throughput scaling}
An arrangement consists of several subcomponents: batch formation, trace
maintenance, and \eg a maintained \texttt{count} operator.
To evaluate throughput scaling, we issue batches of 10,000 updates
at each worker, repeated as soon as each batch is accepted, rather than from a rate-limited open-loop harness.
Figure~\ref{fig:eval:throughput} reports the peak throughputs as the
number of cores (and thus, workers and arrangement shards) grows.
All components scale linearly to 32 workers.

\paragraph{Amortized merging}
The amortized merging strategy is crucial for shared arrangements to achieve
low update latency, but its efficacy depends on setting the right amortization
coefficients.
Eager merging performs the least work overall but can increase tail latency.
Lazy merging performs more work overall, but should reduce the tail latency.
Ideally, \sys's default would pick a good tradeoff between common-case and
tail latencies at different scales.
Figure~\ref{fig:eval:amortized} reports the latency distributions for one and
32 workers, each with three different merge amortization coefficients: the
most eager, \sys's default, and the most lazy possible.
For a single worker, lazier settings have smaller tail latencies, but are
more often in that tail.
For 32 workers, the lazier settings are significantly better, because eager
strategies often cause workers to stall waiting for a long merge at one
worker.
The lazier settings are critical for effective strong scaling, where eager work causes multiple workers to seize up, which matches similar observations about garbage collection at scale~\cite{broom}.
\sys's default setting achieves good performance at both scales.

\paragraph{Join proportionality}
Our arrangement-aware join operator is designed to perform work proportional
to the size of the smaller of the incoming pre-arranged batch and the
state joined against (\subsecref{operators}{op-join}).
We validate this by measuring the latency distributions to
install, execute, and complete new dataflows that join collections of
varying size against a pre-existing arrangement of 10M keys.
The varying lines in Figure~\ref{fig:eval:join_proportionality} demonstrate
that the join work is indeed proportional to the small collection's size,
rather than to the (constant) 10M arranged keys.
This behavior is not possible in a record-at-a-time stream processor, which must at least examine each input record.
This behavior is possible in \sys only because the \texttt{join} operator receives as input pre-arranged batches of updates.
Query deployment in the TPC-H workload would not be fast without this property.




\vspace{1\baselineskip}
\subsection{Baseline performance on reference tasks}
\subseclabel{eval}{broad}

We also evaluate \sys against established prior work to demonstrate that \sys is competitive with and occasionally better than peer systems.
Importantly, these established benchmarks are traditionally evaluated in isolation, and are rarely able to demonstrate the benefits of shared arrangements.
Instead, this evaluation is primarily to demonstrate that \sys does not \emph{lose} baseline performance as compared to other state-of-the-art systems.
Most but not all of the peer systems in this section do maintain private indexed data in operators; this decision alone accounts for some of the gaps.
%

\vspace{.4\baselineskip}
\subsubsection{Datalog workloads}
Datalog is a relational language in which queries are sets of recursively
defined productions, which are iterated from a base set of records until no
new records are produced.
Unlike graph computation, Datalog queries tend to produce and work with
substantially more records than they are provided as input.
Several shared-memory systems for Datalog exist, including LogicBlox,
DLV~\cite{dlv}, DeALS~\cite{deals}, and several distributed systems have
recently emerged, including Myria~\cite{myria}, SociaLite~\cite{socialite},
and BigDatalog~\cite{bigdatalog}.
At the time of writing, only LogicBlox supports decremental updates to
Datalog queries, using a technique called ``transaction
repair''~\cite{DBLP:journals/corr/Veldhuizen14}.
\sys supports incremental and decremental updates to Datalog computations and
interactive top-down queries.

\vspace{.3\baselineskip}
\paragraph{Top-down (interactive) evaluation}
Datalog users commonly specify values in a query, such as
$\textit{reach}(\textit{``david''}, ?)$, to request nodes reachable from a
source node.
The ``magic set'' transformation~\cite{Bancilhon:1985:MSO:6012.15399} rewrites
such queries as bottom-up computations with a new base relation that seeds the
bottom-up derivation with query arguments; the rewritten rules derive facts only
with the participation of some seed record.
\sys, like some interactive Datalog environments, performs this work against
maintained arrangements of the non-seed relations.
We would expect this approach to be much faster than full evaluation, which
batch processors that re-index the non-seed relations (or \sys without shared
arrangements) require.
\tableref{tab:datalog_int} reports \sys's median and maximum latencies for 100
random arguments for three interactive queries on three widely-used benchmark
graphs, and the times for full evaluation of the related query, using 32
workers.
\sys's arrangements mostly reduce runtimes from seconds to milliseconds.
The slower performance for \textrm{sg(x,?)} on \textrm{grid-150} reveals that the transformation is not always beneficial, a known problem with the magic set transform.
%


\begin{figure}[t]
  \small
  \begin{center}
    \begin{tabular}{l|r|r|r|r}
      \textbf{Query} & statistic & tree-11 & grid-150    & gnp1      \\ \hline
      tc(x,?)        & increm., median    &  2.56ms & 346.28ms  & 18.29ms   \\
                     & incremental, max   &  9.05ms & 552.79ms  & 25.40ms   \\
                     & full eval. (no SA) &   0.08s & 6.18s     & 9.45s     \\ \hline
      tc(?,x)        & increm., median    & 15.63ms & 320.83ms  & 15.58ms   \\
                     & incremental, max   & 18.01ms & 541.76ms  & 23.84ms   \\
                     & full evaluation      &   0.08s & 6.18s     & 9.45s     \\ \hline
      sg(x,?)        & increm., median    & 68.34ms & 1075.11ms & 20.08ms   \\
                     & incremental, max   & 95.66ms & 2285.11ms & 26.56ms   \\
                     & full eval. (no SA) & 56.45s  & 0.60s     & 19.85s
    \end{tabular}
    \caption{\sys enables interactive computation of three Datalog
             queries (32 workers, medians and maximums over 100 queries).
             Full evaluation is required without shared arrangements.}
    \label{tab:datalog_int}
  \end{center}
\end{figure}

\vspace{.3\baselineskip}
\paragraph{Bottom-up (batch) evaluation}
\extendedtr{
In \ref{appendix:datalog}, we evaluate \sys relative to distributed and
shared-memory Datalog engines, using their benchmark queries and datasets
(``transitive closure'' and ``same generation'' on trees, grids, and random
graphs).
\tableref{tab:datalog} reports that \sys generally outperforms the
distributed systems, and is comparable to the best shared-memory engine
(DeALS).
}{
In our \theTR, we compare \sys
to distributed and shared-memory Datalog engines, using their benchmarks and
datasets (``transitive closure'' and ``same generation'' on trees,
grids, and random graphs).
Our results show that \sys generally outperforms the distributed systems and is
comparable to the best shared-memory engine (DeALS).
}



\subsubsection{Program Analysis}

Graspan~\cite{graspan} is a system built for static analysis of large code bases,
created in part because existing systems were unable to handle non-trivial
analyses at the sizes required.
%
%
%
Wang et al.\ benchmarked Graspan for two program analyses, \textit{dataflow} and
\textit{points-to}~\cite{graspan}.
The \textit{dataflow} query propagates null assignments along program assignment
edges, while the more complicated \textit{points-to} analysis develops a mutually
recursive graph of value flows, and memory and value aliasing.
We developed a full implementation of Graspan---query parsing, dataflow
construction, input parsing and loading, dataflow execution---in 179 lines of code
on top of \sys.
%

%
Graspan is designed to operate out-of-core, and explicitly manages its data on
disk.
We therefore report \sys measurements from a laptop with only 16 GB of RAM, a
limit exceeded by the \textit{points-to} analysis (which peaks around 30 GB).
The sequential access in this analysis makes standard OS swapping mechanisms
sufficient for out-of-core execution, however.
To verify this, we modify the computation to use 32-bit integers, reducing
the memory footprint below the RAM size, and find that this optimized version
runs only about 20\% faster than the out-of-core execution.

\begin{figure}[t]
  \small
  \centering
  \begin{subfigure}{\columnwidth}
    \centering
    \begin{tabular}{l|r|r|r|r}
      \textbf{System} & cores & \textit{linux} & \textit{psql}  & \textit{httpd} \\
      \hline
      SociaLite       & 4     & OOM               & OOM               & 4 hrs             \\
      Graspan         & 4     & 713.8 min         & 143.8 min         & 11.3 min          \\
      RecStep         & 20    & 430s              & 359s              & 74s               \\ \hline
      \sys            & 1     & 65.8s             & 32.0s             & 8.9s             \\
    \end{tabular}
    \caption{\textit{dataflow} query, \sys on laptop hardware.}
    \label{tab:datalog-df-1}
  \end{subfigure}

  \vspace{6pt}
  \begin{subfigure}{\columnwidth}
    \centering
    \begin{tabular}{l|r|r|r|r}
      \textbf{System} & cores & \textit{linux} & \textit{psql}  & \textit{httpd} \\
      \hline
      RecStep         & 20    & 430s              & 359s              & 74s               \\ \hline
      \sys            & 2     & 53.9s             & 25.5s             & 7.5s             \\
      \sys            & 4     & 34.8s             & 16.3s             & 4.7s             \\
      \sys            & 8     & 24.4s             & 11.2s             & 3.2s             \\
      \sys            & 16    & 20.7s             & 8.7s              & 2.5s             \\

    \end{tabular}
    \caption{\textit{dataflow} query, \sys on server hardware.}
    \label{tab:datalog-df-m}
  \end{subfigure}

  \vspace{6pt}
  \begin{subfigure}{\columnwidth}
    \centering
    \begin{tabular}{l|r|r|r|r}
      \textbf{System} & cores & \textit{linux (kernel only)} & \textit{psql}  & \textit{httpd} \\
      \hline
      \sys (med)      & 1     & 1.05ms            & 143ms             & 18.1ms            \\
      \sys (max)      & 1     & 7.34ms            & 1.21s             & 201ms
    \end{tabular}
    \caption{Times to remove each of the first 1,000 null assignments from the interactive
             top-down \textit{dataflow} query.}
    \label{tab:graspan-topdown}
  \end{subfigure}
\caption{\sys performs well for Graspan~\protect\cite{graspan} \textit{dataflow} query on three graphs.
           SociaLite and Graspan results from Wang et al.~\protect\cite{graspan}; RecStep results from Fan et al.~\protect\cite{recstep}; OOM: out of memory.}
\end{figure}

\begin{figure}[t]
  \small
  \centering
  \begin{subfigure}{\columnwidth}
    \centering
    \begin{tabular}{l|r|r|r|r}
      \textbf{System} & cores & \textit{linux} & \textit{psql}  & \textit{httpd}  \\
      \hline
      SociaLite       & 4     & OOM               & OOM               & $>$ 24 hrs        \\
      Graspan         & 4     & 99.7 min          & 353.1 min         & 479.9 min         \\
      RecStep         & 20    & 61s               & 162s              & 162s              \\ \hline
      \sys            & 1     & 241.0s            & 151.2s            & 185.6s            \\ \hline
      \sys (Opt)     & 1     & 121.1s            & 52.3s             & 51.8s             \\
    \end{tabular}
    \caption{\textit{points-to} analysis, \sys on laptop. \sys (Opt) is an optimized query.}
    \label{tab:datalog-pt-1}
  \end{subfigure}

  \vspace{6pt}
  \begin{subfigure}{\columnwidth}
    \centering
    \begin{tabular}{l|r|r|r|r}
      \textbf{System} & cores & \textit{linux} & \textit{psql}  & \textit{httpd}  \\
      \hline
      RecStep         & 20    & 61s               & 162s              & 162s              \\ \hline
      \sys            & 2     & 230.0s            & 134.4s            & 145.3s            \\
      \sys            & 4     & 142.6s            & 73.3s             & 80.2s             \\
      \sys            & 8     & 86.0s             & 40.9s             & 44.9s             \\
      \sys            & 16    & 59.8s             & 24.0s             & 27.5s             \\ \hline
      \sys (Opt)     & 2     & 125.2s            & 53.1ss            & 46.0s             \\
      \sys (Opt)     & 4     & 89.8s             & 30.8s             & 26.7s             \\
      \sys (Opt)     & 8     & 57.4s             & 18.0s             & 15.1s             \\
      \sys (Opt)     & 16    & 43.1s             & 11.2s             & 9.1s             \\
    \end{tabular}
    \caption{\textit{points-to} analysis, \sys on server. \sys (Opt) is an optimized query.}
    \label{tab:datalog-pt-m}
  \end{subfigure}
  \caption{\sys performs well for Graspan~\protect\cite{graspan} program analyses on three graphs.
           SociaLite and Graspan results from Wang et al.~\protect\cite{graspan}; RecStep results from Fan et al.~\protect\cite{recstep}; OOM: out of memory.}

\end{figure}

%
\tableref{tab:datalog-df-1} and \tableref{tab:datalog-pt-1} show the running times
reported by Wang et al.\ compared to those \sys achieves.
For both queries, we see a substantial improvement (from 24$\times$ to
650$\times$).
%
%
The \textit{points-to} analysis is dominated by the determination of a large
relation (value aliasing) that is used only once.
This relation can be optimized out, as value aliasing is eventually restricted by
dereferences, and this restriction can be performed before forming all value
aliases.
This optimization results in a more efficient computation, but one that reuses some relations several (five) times; the benefits of the improved plan may not be realized by systems without shared arrangements.
\tableref{tab:datalog-pt-1} reports the optimized running times as (Opt).

In \tableref{tab:datalog-df-m} and \tableref{tab:datalog-pt-m} we
also report the runtimes of \sys on these program analysis tasks on server hardware
(with the same hardware configuration as previous sections) and compare them to
RecStep~\cite{recstep}, a state-of-the-art parallel datalog engine.
For all queries, \sys matches or outperforms RecStep running times even
when it is configured to utilize a smaller number of CPU cores.

\paragraph{Top-down evaluation}
Both \textit{dataflow} and \textit{points-to} can be transformed to support
interactive queries instead of batch computation.
\tableref{tab:graspan-topdown} reports the median and maximum latencies to remove
the first 1,000 null assignments from the completed \textit{dataflow} analysis and
correct the set of reached program locations.
While there is some variability, the timescales are largely interactive and
suggest the potential for an improved developer experience.

\subsubsection{Batch graph computation}

We evaluate \sys on standard batch iterative graph computations
on three standard social networks:
LiveJournal, Orkut, and Twitter.
\extendedtr{Results for Twitter are in
\tableref{tab:graphs-tw}. Results for LiveJournal (\tableref{tab:graphs-lj})
and Orkut (\tableref{tab:graphs-ok}) are in Section~\ref{appendix:graphs}.
}{We report results for the largest of the graphs, Twitter, in
\tableref{tab:graphs-tw}; results for LiveJournal and Orkut are available
in our \theTR.
}
%
%
%
Following prior work~\cite{bigdatalog} we use the tasks of single-source reachability (reach), single-source shortest paths (sssp), and undirected connectivity (wcc).
For the first two problems we start from the first graph vertex with any outgoing edges (each reaches a majority of the graph).

We separately report the times required to form the forward and reverse edge arrangements, with the former generally faster than the latter as the input graphs are sorted by the source as in the forward index.
The first two problems require a forward index and undirected connectivity requires indices in both directions, and we split the results accordingly.
We include measurements by Shkapsky et al.~\cite{bigdatalog} for several other systems.
%
We also report running times for simple single-threaded implementations that are not required to follow the same algorithms.
For example, for undirected connectivity we use the union-find algorithm rather than label propagation, which outperforms all systems except \sys at 32 cores.
We also include single-threaded implementations that replace the arrays storing per-node state with hash maps, as they might when the graph identifiers have not been pre-processed into a compact range; the graphs remain densely packed and array indexed.

\sys is consistently faster than the other systems---Myria~\cite{myria},
BigDatalog~\cite{bigdatalog}, SociaLite~\cite{socialite},
GraphX~\cite{graphx}, RecStep~\cite{recstep}, and RaSQL~\cite{rasql}---but
is substantially less efficient than
purpose-written single-threaded code applied to pre-processed graph data.
Such pre-processing is common, as it allows use of efficient static arrays, but
it prohibits more general vertex identifiers or graph updates.
When we amend our purpose-built code to use a hash table instead of an array,
\sys becomes competitive between two and four cores.
These results are independent of shared arrangements, but indicate that \sys's
arrangement-aware implementation does not impose any undue cost on computations
without sharing.

\begin{figure}[t]
  \scriptsize
  \begin{center}
    \begin{tabular}{l|r|r|r|r|r|r}
      \textbf{System}       & cores & index-f &   reach &    sssp & index-r &      wcc \\ \hline
      Single thread         & 1     &       - &  14.89s &  14.89s &       - &   33.99s \\
      \; w/hash map         & 1     &       - & 192.01s & 192.01s &       - &  404.19s \\ \hline
      BigDatalog            & 120   &       - &    125s &    260s &       - &     307s \\
      Myria                 & 120   &       - &    102s &   1593s &       - &    1051s \\
      SociaLite             & 120   &       - &    755s &     OOM &       - &      OOM \\
      GraphX                & 120   &       - &   3677s &   6712s &       - &   12041s \\
      RaSQL                 & 120   &       - &     45s &     81s &       - &     108s \\
      RecStep               & 20    &       - &    174s &    243s &       - &     501s \\ \hline
      \sys                  & 1     & 162.41s & 256.77s & 310.63s & 312.31s &  800.05s \\
      \sys                  & 2     &  99.74s & 131.50s & 159.93s & 164.12s &  417.20s \\
      \sys                  & 4     &  49.46s &  64.31s &  77.27s &  81.67s &  200.28s \\
      \sys                  & 8     &  27.99s &  33.68s &  40.24s &  43.20s &  101.42s \\
      \sys                  & 16    &  18.04s &  17.40s &  20.99s &  24.73s &   51.83s \\
      \sys                  & 32    &  12.69s &  11.36s &  10.97s &  14.44s &   27.48s \\
    \end{tabular}
    \caption{System performance on various tasks on the 42M node, 1.4B edge \texttt{twitter} graph. \sys does not share any arrangements here, but the sharing infrastructure does not harm performance.}
    \label{tab:graphs-tw}
  \end{center}
\end{figure}



\section{Conclusions}
\seclabel{concl}

We described shared arrangements, detailed their design and implementation in \sys, and showed how they yield improved performance for interactive analytics against evolving data.
Shared arrangements enable interactive, incrementally maintained queries against streams by sharing sharded indexed state between operators within or across dataflows.
Multiversioning the shared arrangement is crucial to provide high throughput, and sharding the arrangement achieves parallel speedup.
Our implementation in \sys installs new queries against a stream in milliseconds, reduces the processing and space cost of multiple data\-flows, and achieves high performance on a range of workloads.
In particular, we showed that shared arrangements improve performance for workloads with concurrent queries, such as a streaming TPC-H workload with interactive analytic queries and concurrent graph queries.

Shared arrangements rely on features shared by time-aware data\-flow systems, and the idiom of a single-writer, multiple-reader index should apply to several other popular dataflow systems.
We left undiscussed topics like persistence and availability.
As a deterministic data processor, \sys is well-suited to active-active replication for availability in the case of failures.
In addition, the immutable LSM layers backing arrangements are appropriate for persistence, and because of their inherent multiversioning can be persisted asynchronously, off of the critical path.

\sys~\cite{differential} is the reference open-source implementation of Differential Dataflow, and is in use by several research groups and companies.
%
%


\paragraph{Acknowledgements}
We thank Natacha Crooks, Jon Howell, Michael Isard, and the MIT PDOS group for their valuable feedback, and the many users of \sys who exercised and informed its design. This work was partly supported by Google, VMware, and the Swiss National Science Foundation. Andrea Lattuada is supported by a Google PhD fellowship.

\newpage

{
  \bibliography{refs}
}

\ifextendedtr
\clearpage
\appendix
\else
\fi

\renewcommand{\thesection}{Appendix \Alph{section}}

\ifextendedtr

\extendedtr{

\section{Compaction Theorems}
\label{appendix:theorems}

Let $F$ be an antichain of partially ordered times (a ``frontier''). Writing $f \ge F$ to mean $f$ is greater than some element of $F$, we will say that two times $t_1$ and $t_2$ are ``indistinguishable as of $F$'', written $t_1 \equiv_F t_2$, when
\begin{align*}
t_1 \equiv_F t_2 & \; \textrm{ when } \; \forall_{f \ge F} (t_1 \le f \; \textrm{iff} \; t_2 \le f)
\end{align*}

As performed in the Naiad prototype, we can determine a representative for a time $t$ relative to a set $F$ using the least upper bound ($\wedge$) and greatest lower bound ($\vee$) operations of the lattice of times, taking the greatest lower bound of the set of least upper bounds of $t_i$ and elements of $F$:
\begin{align*}
rep_F(t) & :=  \vee_{f \in F} (t \wedge f)
\end{align*}
The $rep_F$ function finds a representative that is both correct ($t$ and $rep_F(t)$ compare identically to times greater than $F$) and optimal (two times comparing identically to all times greater than $F$ map to the same representative).

The formal properties of $rep_F$ rely on properties of the $\wedge$ and $\vee$ operators, that they are respectively upper and lower bounds of their arguments, and their optimality:
\begin{align*}
b \le a \textrm{ and } c \le a & \; \rightarrow \; (b \wedge c) \le a & (lub) \\
a \le b \textrm{ and } a \le c & \; \rightarrow \; a \le (b \vee c) & (glb)
\end{align*}
In particular, we will repeatedly use that if $t_1 \le (t_2 \vee f)$ for all $f \in F$, then $t_1 \le rep_F(t_2)$.

\begin{theorem}[Correctness]
For any lattice element $t$ and set $F$ of lattice elements, $t \equiv_F rep_F(t)$.
\end{theorem}

\begin{proof}
We prove both directions of the implication in $\equiv_F$ separately, for all $f \ge F$. First assume $t \le f$. By assumption, $f$ is greater than some element $f'$ of $F$, and so $t \wedge f' \le f$ by the ($lub$) property. As a lower bound, $rep_F(t) \le t \wedge f'$ for each $f' \in F$, and by transitivity $ rep_F(t) \le f$. Second assume $rep_F(t) \le f$. Because $t \le (t \wedge f')$ for all $f' \in F$, then $t \le rep_F(t)$ by the ($glb$) property and by transitivity $t \le f$.
\end{proof}

At the same time, $rep_F$ is optimal in that two equivalent elements will be mapped to the same representative.

\begin{theorem}[Optimality]
For any lattice elements $t_1$ and $t_2$ and set $F$ of lattice elements, if $t_1 \equiv_F t_2$ then $rep_F(t_1) = rep_F(t_2)$.
\end{theorem}

\begin{proof}
For all $f \in F$ we have both that $t_1 \le t_1 \wedge f$ and $f \le t_1 \wedge f$, the latter implying that $t_1 \wedge f \ge F$. By our assumption, $t_2$ agrees with $t_1$ on times greater than $F$, making $t_2 \le t_1 \wedge f$ for all $f \in F$. By correctness, $rep_F(t_2)$ agrees with $t_2$ on times greater than $F$, which includes $t_1 \wedge f$ for $f \in F$ and so $rep_F(t_2) \le t_1 \wedge f$ for all $f \in F$. Because $rep_F(t_2)$ is less or equal to each term in the greatest lower bound definition of $rep_F(t_1)$, it is less or equal to $rep_F(t_1)$ itself. The symmetric argument proves that $rep_F(t_1) \le rep_F(t_2)$, which implies that the two are equal (by antisymmetry).
\end{proof}

\section{Relational computations}
\label{appendix:relational}

\tableref{tab:tpch-stream-logical} reports throughput in tuples per second for \cite{hotdog} and \sys on the scale factor 10 TPC-H workload with logical batches of 100,000 elements. \sys has a generally higher and more consistent rate, though is less efficient on lighter queries (q04, q06, and q14); \sys performs no pre-aggregation, which could improve its rates for logically batched queries. For 32 workers, almost all rates are above 10 million updates per second, which correspond to latencies below 10ms between reports.

\begin{figure}
\scriptsize
  \begin{center}
    \begin{tabular}{r|r|r|r}
        query    & DBToaster  & \sys (w=1)  & \sys (w=32) \\ \hline
        TPC-H 01 & 4,372,480    &     9,341,713 &    31,283,993 \\
        TPC-H 02 & 691,260     &     4,388,761 &    29,651,632 \\
        TPC-H 03 & 4,580,003    &    11,049,606 &    37,263,673 \\
        TPC-H 04 & 9,752,380    &     9,046,854 &    30,886,269 \\
        TPC-H 05 & 509,490     &     5,802,513 &    27,952,246 \\
        TPC-H 06 & 101,327,791  &    33,090,863 &    65,335,474 \\
        TPC-H 07 & 646,018     &     7,551,628 &    30,962,626 \\
        TPC-H 08 & 221,020     &     4,949,412 &    28,230,062 \\
        TPC-H 09 & 76,296      &     2,932,421 &    18,119,469 \\
        TPC-H 10 & 5,964,290    &     9,708,371 &    25,037,510 \\
        TPC-H 11 & 591,716     &     1,720,655 &     1,749,464 \\
        TPC-H 12 & 7,469,474    &    11,258,702 &    33,975,983 \\
        TPC-H 13 & 474,765     &     1,446,223 &    16,792,703 \\
        TPC-H 14 & 53,436,252   &    21,908,762 &    38,843,085 \\
        TPC-H 15 & 964        &     5,057,397 &    23,122,916 \\
        TPC-H 16 & 58,721      &     4,435,818 &    23,495,608 \\
        TPC-H 17 & 131,964     &     5,218,907 &    25,888,103 \\
        TPC-H 18 & 971,313     &     5,854,293 &    29,574,347 \\
        TPC-H 19 & 8,776,165    &    22,696,357 &    36,393,109 \\
        TPC-H 20 & 1,871,407    &    16,089,949 &    46,456,453 \\
        TPC-H 21 & 407,540     &     1,968,771 &    10,928,516 \\
        TPC-H 22 & 815,903     &     1,843,397 &    15,233,935 \\
    \end{tabular}
    \caption{Streaming update rates (in tuples per second) for the 22 TPC-H queries at scale factor 10, with logical batching of 100,000 elements at the same time. Elapsed times for DBToaster are for one thread, and are reproduced from~\cite{hotdog}. For \sys we report both one worker and 32 worker rates.}

    \label{tab:tpch-stream-logical}
  \end{center}
\end{figure}

\tableref{tab:tpch-batch} reports elapsed times for \sys applied to the scale-factor 10 TPC-H workload. We also reproduce several measurements from \cite{flare} for other systems. All are executed with a single thread. \sys used as a batch processor is not as fast as the best systems (HyPer and Flare), but is faster than other popular frameworks.
\begin{figure}
\scriptsize
  \begin{center}
    \begin{tabular}{r|r|r|r|r|r}
        query    & Postgres & SparkSQL & HyPer & Flare &   \sys \\ \hline
        TPC-H 01 &   241,404 &    18,219 &   603 &   530 &   7,789 \\
        TPC-H 02 &     6,649 &    23,741 &    59 &   139 &   2,426 \\
        TPC-H 03 &    33,721 &    47,816 & 1,126 &   532 &   5,948 \\
        TPC-H 04 &     7,936 &    22,630 &   842 &   521 &   8,550 \\
        TPC-H 05 &    30,043 &    51,731 &   941 &   748 &  14,001 \\
        TPC-H 06 &    23,358 &     3,383 &   232 &   198 &   1,185 \\
        TPC-H 07 &    32,501 &    31,770 &   943 &   830 &  12,029 \\
        TPC-H 08 &    29,759 &    63,823 &   616 & 1,525 &  19,667 \\
        TPC-H 09 &    64,224 &    88,861 & 1,984 & 3,124 &  27,873 \\
        TPC-H 10 &    33,145 &    42,216 &   967 & 1,436 &   4,559 \\
        TPC-H 11 &     7,093 &     3,857 &   131 &    56 &   1,534 \\
        TPC-H 12 &    37,880 &    17,233 &   501 &   656 &   4,458 \\
        TPC-H 13 &    31,242 &    28,489 & 3,625 & 3,727 &   3,893 \\
        TPC-H 14 &    22,058 &     7,403 &   330 &   278 &   1,695 \\
        TPC-H 15 &    23,133 &    14,542 &   253 &   302 &   1,591 \\
        TPC-H 16 &    13,232 &    23,371 & 1,399 &   620 &   2,238 \\
        TPC-H 17 &   155,449 &    70,944 &   563 & 2,343 &  17,750 \\
        TPC-H 18 &    90,949 &    53,932 & 3,703 &   823 &   9,426 \\
        TPC-H 19 &    29,452 &    13,085 & 1,980 &   909 &   2,444 \\
        TPC-H 20 &    65,541 &    31,226 &   434 &   870 &   4,658 \\
        TPC-H 21 &   299,178 &   128,910 & 1,626 & 1,962 &  29,363 \\
        TPC-H 22 &    11,703 &    10,030 &   180 &   177 &   2,819 \\
    \end{tabular}
    \caption{Elapsed milliseconds for the 22 TPC-H queries at scale factor 10, each using a single core. Elapsed times for the four other systems are reproduced from~\cite{flare}. \sys used as a batch processor is not as fast as the best systems (HyPer and Flare), but is faster than other popular frameworks.}

    \label{tab:tpch-batch}
  \end{center}
\end{figure}


\section{Graph computations}
\label{appendix:graphs}

We now report on the performance of \sys used as a batch graph processor, where the inputs are static collections of edges defining a directed graph. Following \cite{bigdatalog} we use the tasks of single-source reachability (reach), single-source shortest paths (sssp), and undirected connectivity (wcc). For the first two problems we start the process from the first graph vertex with any outgoing edges (each reaches a majority of the graph).

We separately report the times required to form the forward and reverse edge arrangements, with the former generally faster than the latter as the graphs are made available sorted by the source as in the forward index. The reported times for the three problems are then any further time required after the index construction, where the first two problems require a forward index and undirected connectivity requires indices in both directions.

We report times for three graphs:  \texttt{livejournal} in \tableref{tab:graphs-lj}, \texttt{orkut} in \tableref{tab:graphs-ok}, and \texttt{twitter} in \tableref{tab:graphs-tw}. We also reproduce measurements reported in \cite{bigdatalog} for several other systems. We include running times for simple single-threaded implementations that are not required to follow the same algorithms. For example, for undirected connectivity we use the union-find algorithm rather than label propagation, which outperforms all of the measurements reported here. We also include the same times when the single-threaded implementations replace the arrays storing per-node state with hash maps, as they might when the graph identifiers have not be pre-processed to be in a compact range; the graphs remain densely packed and array indexed.

The measurements indicate that \sys is at least as capable as these peer systems at graph processing, often with just a single core. \sys lags well behind graph processors that can rely on dense integer identifiers and direct array access, but appears likely to be competetive with graph processors that cannot rely on dense integer identifiers.
\begin{figure}[t]
  \tiny
  \begin{center}
    \begin{tabular}{l|r|r|r|r|r|r}
      \textbf{System}       & cores & index-f &   reach &    sssp & index-r &      wcc \\ \hline
      Single thread         & 1     &       - &   0.40s &   0.40s &       - &    0.29s \\
      \; w/hash map         & 1     &       - &   4.38s &   4.38s &       - &    8.90s \\ \hline
      BigDatalog            & 120   &       - &     17s &     53s &       - &      27s \\
      Myria                 & 120   &       - &      5s &     70s &       - &      39s \\
      SociaLite             & 120   &       - &     52s &    172s &       - &      54s \\
      GraphX                & 120   &       - &     36s &    311s &       - &      59s \\
      RecStep               & 20    &       - &     14s &     19s &       - &      39s \\ \hline
      \sys                  & 1     &   4.39s &   8.50s &  13.14s &   7.56s &   23.97s \\
      \sys                  & 2     &   2.49s &   4.33s &   6.71s &   4.01s &   12.95s \\
      \sys                  & 4     &   1.39s &   2.31s &   3.58s &   2.06s &    6.29s \\
      \sys                  & 8     &   0.74s &   1.20s &   1.79s &   1.03s &    3.41s \\
      \sys                  & 16    &   0.54s &   0.62s &   0.90s &   0.58s &    1.71s \\
      \sys                  & 32    &   0.55s &   0.51s &   0.59s &   0.41s &    0.90s \\
    \end{tabular}
    \caption{System performance on various tasks on the 4.8M node, 68M edge \texttt{livejournal} graph.}
    \label{tab:graphs-lj}
  \end{center}
\end{figure}

\begin{figure}[t]
  \scriptsize
  \begin{center}
    \begin{tabular}{l|r|r|r|r|r|r}
      \textbf{System}       & cores & index-f &   reach &    sssp & index-r &      wcc \\ \hline
      Single thread         & 1     &       - &   0.46s &   0.46s &       - &    0.52s \\
      \; w/hash map         & 1     &       - &  11.56s &  11.56s &       - &   19.00s \\ \hline
      BigDatalog            & 120   &       - &     20s &     39s &       - &      33s \\
      Myria                 & 120   &       - &      6s &     44s &       - &      57s \\
      SociaLite             & 120   &       - &     67s &    106s &       - &      78s \\
      GraphX                & 120   &       - &     48s &     67s &       - &      53s \\
      RaSQL                 & 120   &       - &     11s &     16s &       - &      19s \\
      RecStep               & 20    &       - &     19s &     25s &       - &      43s \\ \hline
      \sys                  & 1     &  14.02s &  20.33s &  24.65s &  21.27s &   47.79s \\
      \sys                  & 2     &   7.92s &  10.29s &  13.06s &  11.49s &   25.02s \\
      \sys                  & 4     &   4.25s &   5.34s &   6.21s &   5.73s &   12.38s \\
      \sys                  & 8     &   2.37s &   2.68s &   3.34s &   3.03s &    6.29s \\
      \sys                  & 16    &   1.43s &   1.47s &   1.60s &   1.69s &    3.30s \\
      \sys                  & 32    &   1.22s &   1.11s &   0.87s &   1.05s &    1.75s \\
    \end{tabular}
    \caption{System performance on various tasks on the 3M node, 117M edge \texttt{orkut} graph.}
    \label{tab:graphs-ok}
  \end{center}
\end{figure}

\begin{figure}[t]
  \scriptsize
  \begin{center}
    \begin{tabular}{l|r|r|r|r|r|r}
      \textbf{System}       & cores & index-f &   reach &    sssp & index-r &      wcc \\ \hline
      Single thread         & 1     &       - &  14.89s &  14.89s &       - &   33.99s \\
      \; w/hash map         & 1     &       - & 192.01s & 192.01s &       - &  404.19s \\ \hline
      BigDatalog            & 120   &       - &    125s &    260s &       - &     307s \\
      Myria                 & 120   &       - &    102s &   1593s &       - &    1051s \\
      SociaLite             & 120   &       - &    755s &     OOM &       - &      OOM \\
      GraphX                & 120   &       - &   3677s &   6712s &       - &   12041s \\
      RaSQL                 & 120   &       - &     45s &     81s &       - &     108s \\
      RecStep               & 20    &       - &    174s &    243s &       - &     501s \\ \hline
      \sys                  & 1     & 162.41s & 256.77s & 310.63s & 312.31s &  800.05s \\
      \sys                  & 2     &  99.74s & 131.50s & 159.93s & 164.12s &  417.20s \\
      \sys                  & 4     &  49.46s &  64.31s &  77.27s &  81.67s &  200.28s \\
      \sys                  & 8     &  27.99s &  33.68s &  40.24s &  43.20s &  101.42s \\
      \sys                  & 16    &  18.04s &  17.40s &  20.99s &  24.73s &   51.83s \\
      \sys                  & 32    &  12.69s &  11.36s &  10.97s &  14.44s &   27.48s \\
    \end{tabular}
    \caption{System performance on various tasks on the 42M node, 1.4B edge \texttt{twitter} graph.}
    \label{tab:graphs-tw}
  \end{center}
\end{figure}

\section{Datalog computations}
\label{appendix:datalog}

\begin{figure}[ht]
  \tiny
  \begin{center}
    \begin{tabular}{l|r|r|r|r|r|r|r}
      \textbf{System} & cores & \textit{tc}(t) & \textit{tc}(g) & \textit{tc}(r) & \textit{sg}(t) & \textit{sg}(g) & \textit{sg}(r) \\
      \hline
      BigDatalog    & 120 & 49s     & 25s     & 7s      & 53s     & 34s   & 72s     \\
      Spark         & 120 & 244s    & OOM     & 63s     & OOM     & 1955s & 430s    \\
      Myria         & 120 & 91s     & 22s     & 50s     & 822s    & 5s    & 436s    \\
      SociaLite     & 120 & DNF     & 465s    & 654s    & OOM     & 17s   & OOM     \\ \hline
      LogicBlox     & 64  & NR      & 24420s  & 913s    & 58732s  & 326s  & 3363s   \\
      DLV           & 1   & NR      & 13127s  & 9272s   & OOM     & 105s  & 48039s  \\
      DeALS         & 1   & NR      & 148s    & 321s    & 1309s   & 7.6s  & 2926s   \\
      DeALS         & 64  & NR      & 5s      & 12s     & 48s     & 0.35s & 79s     \\ \hline
      \sys          & 1   &  98.26s & 132.23s & 210.25s & 1210.78s& 4.43s & 482.91s \\
      \sys          & 2   &  53.42s &  68.13s & 111.98s & 652.74s & 2.76s & 253.80s \\
      \sys          & 4   &  27.85s &  34.42s &  57.69s & 325.24s & 1.63s & 125.00s \\
      \sys          & 8   &  15.37s &  17.97s &  30.90s & 173.96s & 1.06s &  66.10s \\
      \sys          & 16  &   9.63s &   9.74s &  16.66s &  93.47s & 0.69s &  35.44s \\
      \sys          & 32  &   7.18s &   6.18s &   9.45s &  56.45s & 0.60s &  19.85s
    \end{tabular}
    \caption{System performance on various Datalog problems and graphs. Times for the first four systems are reproduced from~\cite{bigdatalog}. NR indicates the measurements were not reported, DNF indicates a run that lasted more than 24 hours, and OOM indicates the system terminated due to lack of memory.}
    \label{tab:datalog}
  \end{center}
\end{figure}

\tableref{tab:datalog} reports elapsed seconds first for distributed systems, then for single-machine systems, and then for \sys at varying numbers of workers. The workloads are ``transitive closure'' (\textit{tc}) and ``same generation'' (\textit{sg}) on supplied graphs that are trees (\textit{t}), grids (\textit{g}), and random graphs (\textit{r}).

\sys is generally competitive with the best of the specialized Datalog systems (here: DeALS), and generally out-performs the distributed data processors. BigDatalog competes well on transitive closure due to an optimization for linear queries where it broadcasts its input dataset to all workers; this works well with small inputs, as here, but is not generally a robust strategy.

}{}

\else
\fi

\end{document}